%% file: main.tex
\pdfoutput=1

\documentclass[sigplan,10pt]{acmart}

\setcopyright{none}
\renewcommand\footnotetextcopyrightpermission[1]{}
\usepackage{booktabs}
\usepackage{array}
\usepackage{pifont}
\usepackage{microtype}
\usepackage{xurl}
\usepackage{xspace}
\usepackage{listings}
\usepackage{algorithm}
\usepackage{algpseudocode}
\algrenewcommand{\algorithmiccomment}[1]{\hfill{\scriptsize\textcolor{gray!60!black}{$\triangleright$ #1}}}
\usepackage{tikz}
\usetikzlibrary{fit, positioning, arrows.meta, shapes.geometric, decorations.pathreplacing, calc, patterns}
\usepackage{pgfplots}
\pgfplotsset{compat=1.18}

\colorlet{fmpale}{black!6}
\colorlet{fmwin}{orange!50}
\colorlet{fmreuse}{green!45!black!46!white}
\colorlet{fmexec}{black!55}

\newcommand{\sys}{DynBranch\xspace}
\newcommand{\wl}[1]{\emph{#1}}
\newcommand{\cnum}[1]{\tikz[baseline=-0.5ex]{%
  \node[draw,circle,line width=0.4pt,inner sep=0pt,minimum size=1.10em,
        font=\fontsize{0.78em}{0.78em}\selectfont]{#1};}}

\AtBeginDocument{%
  \setlength{\abovecaptionskip}{6pt}%
  \setlength{\belowcaptionskip}{0pt}%
  \setlength{\textfloatsep}{10pt plus 3pt minus 3pt}%
  \setlength{\floatsep}{10pt plus 3pt minus 2pt}%
  \setlength{\intextsep}{10pt plus 3pt minus 2pt}%
  \setlength{\dbltextfloatsep}{10pt plus 3pt minus 3pt}%
  \setlength{\dblfloatsep}{10pt plus 3pt minus 2pt}%
}

\begin{document}

\title{\sys: Speculative Subgraph Reuse for Dynamic Agentic LLM Serving}

\author{Junyi Shen}
\affiliation{%
  \institution{National University of Singapore}
  \country{Singapore}}
\email{j1shen@comp.nus.edu.sg}

\author{Noppanat Wadlom}
\affiliation{%
  \institution{National University of Singapore}
  \country{Singapore}}
\email{noppanat@comp.nus.edu.sg}

\author{Zhengyuan Su}
\affiliation{%
  \institution{National University of Singapore}
  \country{Singapore}}
\email{suzhengy@comp.nus.edu.sg}

\author{Yao Lu}
\affiliation{%
  \institution{National University of Singapore}
  \country{Singapore}}
\email{luyao@comp.nus.edu.sg}

\renewcommand{\shortauthors}{Shen et al.}

\begin{abstract}
\input{sections/abstract}
\end{abstract}

\maketitle
\pagestyle{plain}

\input{sections/intro}
\input{sections/background}
\input{sections/design}
\input{sections/implementation}
\input{sections/eval}
\input{sections/related}
\input{sections/conclusion}

\bibliographystyle{ACM-Reference-Format}
\bibliography{refs}

\appendix
\input{supplement/appendix}
\setcounter{table}{0}
\input{supplement/evidence}

\end{document}

%% file: sections/abstract.tex
Agentic LLM workflows decide their execution paths at runtime.
Downstream computation may be predictable, or may have run before, yet
it cannot begin until the model or the user resolves the branch. We call
this serialization the \emph{branch-resolution barrier}. Caching alone does
not hide it: the key that identifies a reusable result is not known
until then.

In this paper, we propose \sys{}, which makes an unresolved branch
addressable before it resolves. Its stable \emph{coordinate} lets
candidate subgraphs run during resolution and completed subgraph results be
reused across later requests. A two-level
controller admits this work when its expected benefit exceeds the
load price. \sys{} sits at the model-API boundary and
requires no changes to agent harnesses or model execution engines.

Across four agentic workloads with Qwen3-32B on 4$\times$ H200 GPUs,
\sys{} reduces mean latency by up to $32\%$ over each workload's
strongest prior system and by $46$--$66\%$ against a no-reuse floor,
while preserving workflow results. The benefit persists across backbone
families and on a commodity Qwen3-8B/RTX 4090 deployment.

%% file: sections/intro.tex
\section{Introduction}
\label{sec:intro}

LLM applications are moving from workflows whose model and tool calls
are fixed before execution~\cite{parrot,llmcompiler} to agentic ones
that decide each next step as they run, by the model or by the
user~\cite{react,autogen,sweagent}. Measured agent workloads show tens to hundreds
of model and tool calls per task~\cite{tracelab,autellix}, so these
decision windows repeatedly enter the critical path. Downstream work waits on a
choice the model has yet to decode: the handler a request needs, the
sub-task a sub-agent runs, the entity a retrieval fetches. Each is
typically followed by model inference. For example, a ReAct agent answering
multi-hop questions decodes the entity it will search
next~\cite{react}, and coding agents such as
Claude Code and Codex show both: they decode the next tool call, and at
interactive decision points they wait for the user to pick from a short
list of suggested actions~\cite{claudecode,codexcli}.

We call this serialization the \emph{branch-resolution barrier}:
downstream computation that may be predictable, or may have run before,
waits for the model output or user choice that determines the branch.
In measured coding, research, and scientific agents, $45$--$57\%$ of
end-to-end latency is downstream time serialized after model output~\cite{paste}.
This creates two opportunities.
Starting likely downstream computation before resolution can hide part of
its latency on the first occurrence (Figure~\ref{fig:coproduced-windows}). Reusing a validated result can avoid
that computation on later occurrences.

\begin{figure}[t]
\centering
\resizebox{\columnwidth}{!}{
\begin{tikzpicture}[
  font=\normalsize,
  box/.style={draw=black, line width=0.5pt, fill=white, minimum height=14pt, inner xsep=4pt, inner ysep=2pt, align=center},
  win/.style={box, fill=fmwin},
  stage/.style={box},
  keptbr/.style={box, densely dashed},
  winbr/.style={box, line width=1.1pt, fill=fmpale},
  losebr/.style={box, densely dotted, draw=black!60, text=black!55},
  fl/.style={-{Latex[length=4pt, width=3pt]}, line width=0.5pt, shorten >=1pt, shorten <=1pt},
  specfl/.style={fl, densely dashed},
  lbl/.style={font=\normalsize\bfseries},
  status/.style={font=\normalsize, text=black!55, anchor=west},
]
\node[lbl, anchor=west] at (-0.2, 1.85) {(a) human-wait window --- indep.\ top-$K$};
\node[stage] (offer) at (0.7, 0.3) {offer\\$A/B/C$};
\node[win, minimum width=66pt] (wait) at (3.0, 0.3) {user deliberates\\{\normalsize ($\approx$2.7\,s, idle)}};
\node[winbr] (ba) at (5.25, 0.98) {$B_A$};
\node[keptbr] (bb) at (5.25, 0.3) {$B_B$};
\node[losebr] (bc) at (5.25, -0.38) {$B_C$};
\node[status, anchor=south] at ([yshift=1pt]ba.north) {used now};
\node[status, align=left] at ([xshift=1pt]bb.east) {reusable};
\node[status] at ([xshift=1pt]bc.east) {discarded};
\node[stage] (click) at (7.78, 0.3) {click\\$=A$};
\draw[fl] (offer) -- (wait);
\draw[specfl] (wait.east) -- (ba.west);
\draw[specfl] (wait.east) -- (bb.west);
\draw[specfl] (wait.east) -- (bc.west);
\draw[fl, line width=1.1pt] (click.north) |- ([yshift=1pt]ba.east);
\begin{scope}[yshift=-2.88cm]
\node[lbl, anchor=west] at (-0.2, 1.75) {(b) model-decode window --- exclusive};
\node[stage] (req) at (0.7, 0.425) {request};
\node[win, minimum width=66pt] (dec) at (3.0, 0.425) {LLM decodes\\{\normalsize ($\approx$4.6\,s)}};
\node[winbr] (h1) at (5.25, 0.88) {$B_{h_1}$};
\node[losebr] (h2) at (5.25, -0.03) {$B_{h_2}$};
\node[status, anchor=south] at ([yshift=1pt]h1.north) {commit};
\node[status] at ([xshift=1pt]h2.east) {cancel};
\node[stage] (emit) at (7.78, 0.3) {emit\\$h_1$};
\draw[fl] (req) -- (dec);
\draw[specfl] (dec.east) -- (h1.west);
\draw[specfl] (dec.east) -- (h2.west);
\draw[fl, line width=1.1pt] (emit.north) |- ([yshift=1pt]h1.east);
\end{scope}
\end{tikzpicture}}
\caption{Predicted subgraphs fill during the shaded resolution window;
each candidate's status is printed beside it. (a)~independent top-$K$.
(b)~exclusive route.}
\Description{Two timelines compare branch promotion policies.
In the human-wait case, the offered candidates A, B, and C fill
independently: A is used after the click, completed B is kept as a reusable shadow,
and unfinished C is discarded. In the LLM-decode case, emission
of h1 commits its branch and cancels h2.}
\label{fig:coproduced-windows}
\end{figure}

Existing systems provide result caching~\cite{gptcache,helium} and
agent-level speculation~\cite{specactions,dsap,spagent} as separate
mechanisms. Caches retain completed responses and operator results, while
speculative agents consume or discard predicted work within the current
request. Prefix caches reuse KV state~\cite{sglang}, and dataflow schedulers
overlap work once its inputs are known~\cite{parrot,teola,llmcompiler,conveyor}.
These systems do not jointly support reusing a subgraph result across
later requests, starting a branch before it resolves, and pricing that
work against serving load. Each is hard in a different way.
\emph{C1, shared-subgraph reuse}: the unit that recurs is a sub-DAG
rather than a whole request, and moving the boundary there must not
weaken exact matching or freshness. \emph{C2, early fill during
resolution}: the name that authorizes reuse does not exist until the branch
resolves, yet the work must start before it does. \emph{C3, load-priced
admission}: speculative work draws on the same slots and KV as the
request it accelerates, so what helps at low load hurts at high load.

\looseness=-1
In this paper, we propose \sys{}, which makes an unresolved branch
addressable before it resolves (\S\ref{sec:design}). For C1, it
records a completed branch subgraph together with the input that
produced it and the external reads it made. A later request consumes
the entry only when
its input matches exactly and those reads are still fresh; a write to
any of them invalidates it. For C2, the workflow template provides a
\emph{coordinate}: the identity of a decision point in the template,
such as the router stage of a routing workflow, known before the branch
resolves in our target workloads (\S\ref{sec:discussion}). Predicted candidate subgraphs run as
\emph{fills} during the resolution window; completed results can also
serve later requests. Fills remain invisible until branch selection
and an exact, fresh \emph{canonical} (non-speculative) demand authorize
their use, so wrong
predictions cannot change workflow results. For
C3, a two-level controller decides whether to generate candidates
and whether to execute each fill, admitting either when expected
benefit exceeds the load price over the fill's execution horizon. That
price reads the engine's current decode and KV occupancy, so it rises as
speculative and canonical work accumulate.

We implement \sys{} at the model-API boundary between agent harnesses
and unmodified SGLang or vLLM engines~\cite{sglang,vllm}, an interface
that covers both framework-driven
workflows~\cite{autogen,langgraph,agentscope} and text-based coding
agents~\cite{dsh,claudecode,codexcli}. 
With Qwen3-32B on 4$\times$ H200 GPUs, \sys{} reduces mean latency by
up to $32\%$ over each workload's strongest prior system on the same
GPUs and by
$46$--$66\%$ against a no-reuse floor, while preserving workflow
results. With reuse disabled, early fill alone reduces mean latency by
$8.5$--$34.4\%$ on model-driven paths and user-turn latency by
$46.6$--$47.7\%$ on \wl{HCI} (\S\ref{sec:eval:c1}). The benefit persists across backbone families and on a commodity
Qwen3-8B/RTX 4090 deployment.

We make three contributions:

\begin{itemize}
  \item We identify the branch-resolution barrier and quantify its cost
    on production traces and four agentic workloads (\S\ref{sec:bg}).
  \item We propose a branch-state abstraction connecting speculative
    execution with cross-request subgraph reuse, promoted on exact
    match and invalidated by writes
    (\S\ref{sec:design:smg}--\S\ref{sec:design:lifecycle}).
  \item We design a two-level online controller that prices candidate
    generation and fill execution against benefit and load, learning
    from delayed feedback
    (\S\ref{sec:design:admission},
    \S\ref{sec:eval:c1}--\S\ref{sec:eval:c1c2-ablation}).
\end{itemize}

%% file: sections/background.tex
\section{Background and Motivation}
\label{sec:bg}
\label{sec:bg:challenges}

An agentic workflow interleaves model calls, tool calls, and user turns,
and at many steps the next one depends on the current step's output: the
model names the tool or sub-task to run, or the user picks among offered
actions. We call each model or tool call, executed on
its rendered input (the prompt or tool arguments after upstream outputs
are substituted), a \emph{stage}, and the stages downstream of a
decision point its \emph{subgraph}. Figure~\ref{fig:funnel} narrows from
all agent work to the part serialized behind such decisions, then to the
part that can start before a decision resolves (C2) and the part a later
request can reuse (C1), all drawn from capacity that canonical
requests also need (C3).

\begin{figure}[t]
\centering
\resizebox{\columnwidth}{!}{
\begin{tikzpicture}[
  x=0.9cm, y=1.15cm,
  font=\fontsize{12}{14}\selectfont,
  stg/.style={anchor=west, font=\fontsize{11.7}{13.7}\selectfont\bfseries},
  chip/.style={draw=black, line width=0.5pt, inner sep=1.4pt,
               font=\fontsize{11.7}{13.7}\selectfont\bfseries, fill=white, anchor=west},
  wild/.style={anchor=west, font=\fontsize{12}{14}\selectfont},
  ours/.style={anchor=west, font=\fontsize{12}{14}\selectfont\itshape, text=black!55},
  down/.style={-{Latex[length=4pt, width=3pt]}, black, line width=0.5pt},
]
\foreach \xl/\xr/\yb/\xll/\xrr/\yt in {%
    0.000/4.600/-0.32/0.325/4.275/-0.63,
    0.325/4.275/-1.27/0.650/3.950/-1.58,
    0.650/3.950/-2.22/0.950/3.650/-2.53} {
  \path[fill=black!4] (\xl,\yb) -- (\xr,\yb) -- (\xrr,\yt) -- (\xll,\yt) -- cycle;
  \draw[black!55, line width=0.5pt] (\xl,\yb) -- (\xll,\yt);
  \draw[black!55, line width=0.5pt] (\xr,\yb) -- (\xrr,\yt);
}
\filldraw[fill=fmpale, draw=black!70, line width=0.5pt]
                                        (0.000,0.32) rectangle (4.600,-0.32);
\filldraw[fill=fmwin, draw=orange!85!black, line width=0.5pt]
                                        (0.325,-0.63) rectangle (4.275,-1.27);
\filldraw[fill=orange!83, draw=orange!85!black, line width=0.5pt]
                                        (0.650,-1.58) rectangle (3.950,-2.22);
\filldraw[fill=green!45!black!75!white, draw=green!30!black, line width=0.5pt]
                                        (0.950,-2.53) rectangle (3.650,-3.17);
\node[stg, anchor=center] at (2.30,0) {Agent work};
\node[chip] at (0.959,-0.95) {F1--F3};
\node[stg] at (2.35,-0.95) {Serial};
\node[chip] at (1.397,-1.90) {C2};
\node[stg] at (2.157,-1.90) {Fillable};
\node[chip] at (1.282,-2.85) {C1};
\node[stg] at (2.042,-2.85) {Reuse};
\node[wild] at (4.80,0.19) {$82$--$318$ steps/task~\cite{tracelab,osworld2}};
\node[ours] at (4.80,-0.19) {$50$ paths / $52$ requests};
\node[wild] at (4.80,-0.76) {$45$--$57\%$ serial time~\cite{paste}};
\node[ours] at (4.80,-1.14) {$26$--$43\%$ overlap};
\node[wild] at (4.80,-1.71) {$59$--$94\%$ coverage~\cite{patil2024gorilla}};
\node[ours] at (4.80,-2.09) {top-$1$ prior $29/30$ turns};
\node[wild] at (4.80,-2.66) {$54$--$62\%$ prefix hits~\cite{kvcachewild}};
\node[ours] at (4.80,-3.04) {$26\%$ sub-tasks recur};
\draw[fill=black!4, draw=black!55, line width=0.4pt]
                                        (0.00,-3.62) rectangle (10.25,-4.46);
\node[chip] at (0.18,-4.04) {C3};
\node[font=\fontsize{12}{14}\selectfont, anchor=west] at (1.05,-3.85) {Spare capacity,};
\node[font=\fontsize{12}{14}\selectfont, anchor=west] at (1.05,-4.23) {Priced by load};
\node[wild] at (4.80,-3.85) {$\approx$$81\%$ SM idle~\cite{rose}};
\node[wild] at (4.80,-4.23) {$>$$50\%$ unused compute~\cite{memgap}};
\end{tikzpicture}}
\caption{Opportunities for early execution (C2) and cross-request reuse
(C1), subject to available serving capacity (C3). Upright figures are
from prior work; italic figures are our own measurements with
speculation disabled.}
\Description{A funnel narrowing downward through four stages --- agent
work, the serialized barrier, fillable, reusable --- each annotated with
a prior-work measurement and the corresponding measurement on our own
workloads, above a bar marking the spare serving capacity that
admission prices.}
\label{fig:funnel}
\end{figure}
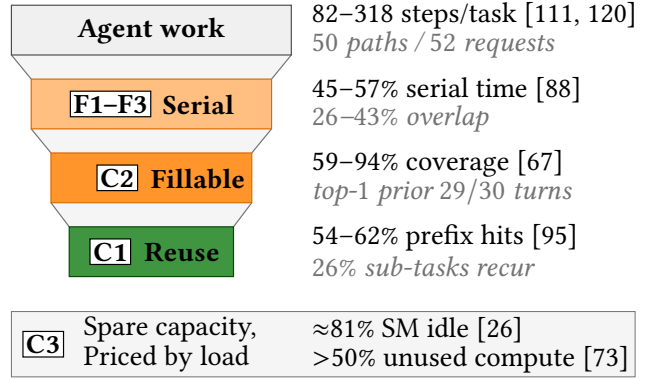

\begin{figure*}[!t]
\centering
\resizebox{\textwidth}{!}{
\begin{tikzpicture}[
  font=\normalsize,
  box/.style={draw=black, line width=0.5pt, fill=white, minimum height=14pt, inner xsep=4pt, inner ysep=2pt, align=center},
  stg/.style={box},
  win/.style={box, fill=fmwin},
  spec/.style={box, densely dashed},
  reuse/.style={box, fill=fmreuse, draw=green!30!black, line width=0.9pt},
  ax/.style={font=\normalsize, text=black!55, align=center},
  ttl/.style={font=\normalsize\bfseries, align=center},
  fl/.style={-{Latex[length=4pt, width=3pt]}, line width=0.5pt, shorten >=1pt, shorten <=1pt},
  specfl/.style={fl, densely dashed},
]
\begin{scope}
\node[ttl] at (1.886,2.65) {\wl{Routing}};
\node[win] (r1) at (0.648,1.5) {router\\decodes};
\node[spec] (h1) at (2.142,2.1) {$h_1$};
\node[spec] (h2) at (2.142,1.5) {$h_2$};
\node[spec] (h3) at (2.142,0.9) {$h_3$};
\node[spec] (ex) at (3.380,2.1) {exec};
\draw[specfl] (r1)--(h1);
\draw[specfl] (r1)--(h2);
\draw[specfl] (r1)--(h3);
\draw[specfl] (h1)--(ex);
\node[ax] at (1.886,0.05) {decode $\cdot$ bounded-$K$\\exclusive $\cdot$ chain reuse};
\end{scope}
\begin{scope}[xshift=4.559cm]
\node[ttl] at (2.034,2.65) {\wl{ReAct}};
\node[win] (t1) at (0.441,1.5) {think};
\node[spec] (rv) at (2.104,1.5) {retrieve\\\emph{entity?}};
\node[stg] (rd) at (3.696,1.5) {read};
\draw[specfl] (t1)--(rv);
\draw[fl] (rv)--(rd);
\draw[fl] (rd.north) -- ++(0,0.42) -| (t1.north);
\node[ax] at (2.034,0.05) {decode $\cdot$ \emph{open-vocab}\\exclusive $\cdot$ \emph{per-hop} loop};
\end{scope}
\begin{scope}[xshift=9.414cm]
\node[ttl] at (1.956,2.65) {\wl{Sub-Agent}};
\node[win] (o0) at (0.366,1.5) {plan};
\node[spec] (s1) at (1.568,2.1) {$S_1$};
\node[spec] (s2) at (1.568,0.9) {$S_2$};
\node[reuse] (rr) at (3.157,2.1) {ret$\to$read};
\node[stg] (sy) at (3.157,0.9) {synth};
\draw[specfl] (o0)--(s1);
\draw[specfl] (o0)--(s2);
\draw[specfl] (s1)--(rr);
\draw[fl] (rr)--(sy);
\node[ax] at (1.956,0.05) {decode $\cdot$ bounded-$K$\\exclusive $\cdot$ \emph{cross-request}};
\end{scope}
\begin{scope}[xshift=14.113cm]
\node[ttl] at (1.746,2.65) {\wl{HCI}};
\node[stg] (m0) at (0.435,1.5) {offer};
\node[win] (m1) at (1.914,1.5) {user\\waits};
\node[spec] (ga) at (3.226,2.1) {$B_A$};
\node[spec] (gb) at (3.226,1.5) {$B_B$};
\node[spec] (gc) at (3.226,0.9) {$B_C$};
\draw[fl] (m0)--(m1);
\draw[specfl] (m1)--(ga);
\draw[specfl] (m1)--(gb);
\draw[specfl] (m1)--(gc);
\node[ax] at (1.746,0.05) {\emph{wait} $\cdot$ bounded-$K$\\\emph{indep.\ top-$K$} $\cdot$ \emph{cross-turn}};
\end{scope}
\end{tikzpicture}}
\caption{Four representative workloads. Shaded boxes are resolution
windows; dashed boxes are candidate work; the heavy-outlined box is a
subgraph that recurs across requests. Italic labels differ from
\wl{Routing}.}
\Description{Four workload panels. Routing resolves a bounded handler
choice during router decode and then runs the selected executor; ReAct
predicts an open-vocabulary retrieval entity inside a per-hop loop;
Sub-Agent plans which sub-agent to run, and the chosen one's
retrieve-to-read subtask can recur across requests before the synthesis
step; HCI fills independently offered continuations during user wait.}
\label{fig:workloads}
\end{figure*}
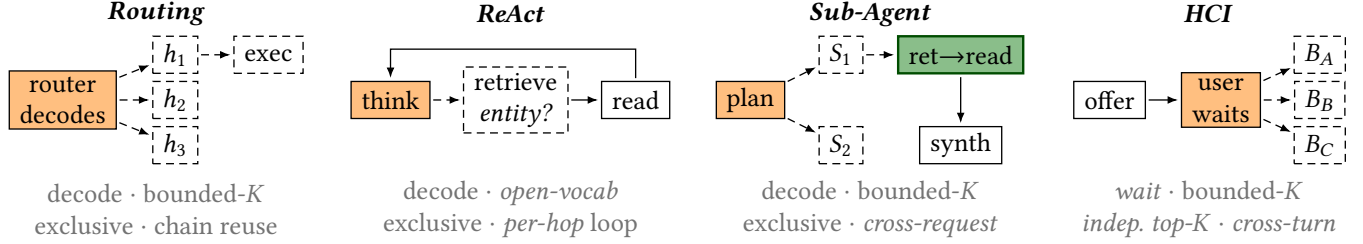

\vspace{0.05in}\noindent\textbf{C1: shared-subgraph reuse.} GPTCache and agentic plan
caching match whole responses or
plans~\cite{gptcache,gptsemcache,plancache}, while SGLang, Pythia, PBKV,
and KVFlow retain prefix-KV state but still execute the
decode~\cite{sglang,pythia,pbkv,kvflow}. Both anchor at the request
boundary, at the whole request or its first token, so neither can
name a sub-DAG that recurs inside requests differing
elsewhere. Its retrieval and tool results also depend on external reads,
so exact matching is not enough: writes to recorded dependencies must
invalidate them (\S\ref{sec:design:cascade}).

\vspace{0.05in}\noindent\textbf{C2: early fill during resolution.} KVFlow, Pythia, and PBKV
prefetch or retain prefix state for a declared or predicted next
stage~\cite{kvflow,pythia,pbkv}, but the stage itself still runs only
after resolution. Speculative Actions executes a predicted API call
early, \mbox{SPAgent} a search action~\cite{specactions,spagent}, DSP a
chain of planning steps~\cite{dsap}, and Speculative Interaction Agents
dependent tool calls alongside model reasoning~\cite{spectool}. PASTE,
the closest, isolates a predicted tool result until a matching
authoritative invocation commits it~\cite{paste}. All five, however,
can use a speculative result only in the request that produced it: the
result is checked against that request's next step and used or
discarded there, so even a correct result is lost when the request
ends. Early fill therefore needs results that stay
isolated until a demand selects them, yet outlive the request that
produced them, which requires naming them before the branch resolves
(\S\ref{sec:design:lifecycle}).

\vspace{0.05in}\noindent\textbf{C3: load-priced admission.} Fills run on
the same GPUs as the requests they accelerate, so each unused fill delays
canonical traffic (Figure~\ref{fig:cobatch}b), more so at high
load~\cite{specqueue}. PASTE, SPAgent, TurboSpec,
and Nightjar limit speculative work under
load~\cite{paste,spagent,turbospec,nightjar}, but each scores a unit once,
before it runs. A branch incurs two costs, running the predictor and then
executing the admitted candidates, and whether either paid off is known
only when the branch resolves, often seconds later. Admission therefore
needs a cost that tracks current load and a benefit estimate corrected by
these delayed outcomes (\S\ref{sec:design:admission}).

\vspace{0.05in}\noindent\textbf{What a workload must supply.}
\label{sec:bg:wall-clock}
First, the work needs an identity independent of runtime content, since
conventional memoization names a result only once its inputs are
known~\cite{gptcache,plancache}; a workflow template supplies one by
naming each decision point, though the model still produces the selection or tool arguments
(\S\ref{sec:discussion}). Second, the stage must be safe to run twice:
no external side effects, same result for the same input and dependency
versions. Greedy inference with batch-invariant
kernels~\cite{he2025nondeterminism} over an immutable model and read-only
retrieval over a versioned snapshot qualify; stages that sample or mutate
external state run only after resolution (\S\ref{sec:design:lifecycle}).
\S\ref{sec:bg:trace-shape} describes four workloads that meet both, and
\S\ref{sec:bg:barrier} measures what each challenge is worth on them.

\subsection{Dynamic Agentic Workloads}
\label{sec:bg:trace-shape}

Agentic workflows determine downstream steps at
runtime~\cite{react,autogen,claudecode}, and the decisions are frequent
and diverse~\cite{tracelab,codeagentbehav}: agent sessions and tasks run
$82$--$318$ model steps~\cite{tracelab,osworld2}, and we measure $50$
distinct paths in $52$ requests, so every run opens tens to hundreds of resolution
windows in which to speculate.

\vspace{0.05in}\noindent\textbf{A workload taxonomy.} Dynamic branches vary along four
dimensions that set the benefit of early execution and the life of its
results: the \emph{resolution window} (model decode or human wait), the
\emph{candidate space} (bounded or open vocabulary), the
\emph{promotion semantics} (one selected winner or multiple valid
results), and the \emph{reuse horizon} (within a chain, across turns,
or across requests). Figure~\ref{fig:workloads} places four workloads on
them. \wl{Routing} is a tool-use
workflow over BFCL function calls~\cite{patil2025bfcl}, \wl{ReAct}~\cite{react} a multi-hop
question-answering agent~\cite{hotpotqa,musique},
\wl{Sub-Agent}~\cite{autogen} an orchestrator that dispatches tasks to
sub-agents, and \wl{HCI} an assistant that offers actions and waits for
the user to pick one~\cite{doherty,klm}.

\subsection{The Branch-Resolution Barrier}
\label{sec:bg:barrier}

\noindent\textbf{Most task latency is serialized.} In coding, research, and
scientific agents, downstream execution takes $45$--$57\%$ of task
latency on the serial model--tool path~\cite{paste}. At branch $k$, work
of duration $T_k^{\mathrm{exec}}$ waits for a resolution window
$T_k^{\mathrm{gen}}$, which for \wl{HCI} is the user's sampled think
time~\cite{doherty,klm}. Early execution can overlap at
most $\min(T_k^{\mathrm{gen}},T_k^{\mathrm{exec}})$; summed along the
critical path and normalized by wall-clock, this gives the workload's
\emph{overlap share}, the most C2 can hide.

Figure~\ref{fig:barrier} shows three observations.
\emph{F1, substantial latency can be overlapped}: the decode-window
workloads can overlap $26$--$43\%$ of wall-clock, and \wl{HCI}'s post-click
execution fits entirely within the modeled think window.
\emph{F2, the window cannot easily be shortened}: on \wl{Routing}, the
resolution decode and the downstream prefill carry similar token volumes
($270$ vs.\ $263$ at p50), yet only the former is emitted
autoregressively~\cite{vllm,sarathi}, and bounding that decode can change
the answer~\cite{s1scaling,cotfaithful}.
\emph{F3, window duration limits the overlap}: at p50, outside \wl{HCI},
the resolution window is $0.36$--$0.77\times$ as long as the downstream
execution.

\begin{figure}[t]
\centering
\begin{tikzpicture}
\begin{axis}[
  width=0.93\columnwidth, height=0.52\columnwidth,
  xbar, xmin=0, xmax=21,
  symbolic y coords={HCI,Sub-Agent,ReAct,Routing},
  ytick=data, y dir=normal,
  xlabel={seconds (p50, whiskers p25--p75)},
  bar width=5.5pt, enlarge y limits=0.18, area legend,
  xmajorgrids, grid style={gray!20},
  tick label style={font=\normalsize}, yticklabel style={font=\normalsize\itshape}, label style={font=\normalsize},
  legend style={at={(0.5,-0.40)}, anchor=north, legend columns=2,
    font=\normalsize, draw=none, fill=none,
    /tikz/every even column/.append style={column sep=8pt}},
  clip=false,
]
\addplot[xbar, fill=fmwin, draw=orange!85!black,
  error bars/.cd, x dir=both, x explicit,
  error bar style={orange!70!black, line width=0.4pt},
  error mark options={orange!70!black, mark size=1.5pt, line width=0.4pt}]
  coordinates {
    (4.6,Routing)   -= (0.8,0) += (0.7,0)
    (3.2,Sub-Agent) -= (0.7,0) += (1.0,0)
    (6.4,ReAct)     -= (1.3,0) += (1.0,0)
    (2.65,HCI)      -= (0.76,0) += (1.06,0)
  };
\addlegendentry{$T^{\mathrm{gen}}$ (window)}
\addplot[xbar, fill=fmexec, draw=black!80,
  error bars/.cd, x dir=both, x explicit,
  error bar style={black!85, line width=0.4pt},
  error mark options={black!85, mark size=1.5pt, line width=0.4pt}]
  coordinates {
    (12.9,Routing)   -= (2.4,0) += (1.9,0)
    (4.8,Sub-Agent)  -= (1.1,0) += (1.2,0)
    (8.3,ReAct)      -= (1.9,0) += (1.6,0)
    (0.87,HCI)       -= (0.0,0) += (0.0,0)
  };
\addlegendentry{$T^{\mathrm{exec}}$ (downstream)}
\node[font=\normalsize, anchor=east, yshift=-5.5pt] at (axis cs:20.9,Routing) {overlap $0.26$};
\node[font=\normalsize, anchor=east, yshift=-5.5pt] at (axis cs:20.9,Sub-Agent) {overlap $0.43$};
\node[font=\normalsize, anchor=east, yshift=-5.5pt] at (axis cs:20.9,HCI) {overlap $\approx$$1.00^{\dagger}$};
\node[font=\normalsize, anchor=east, yshift=-5.5pt] at (axis cs:20.9,ReAct) {overlap $0.32$};
\end{axis}
\end{tikzpicture}
\caption{Resolution windows and downstream execution times; bars show
p50, whiskers p25--p75. $^{\dagger}$Normalized by post-click execution,
not wall-clock.}
\Description{Horizontal paired bars compare the median branch-resolution
window and downstream execution time for Routing, ReAct, Sub-Agent, and
HCI, with interquartile whiskers and each workload's overlap share.}
\label{fig:barrier}
\end{figure}
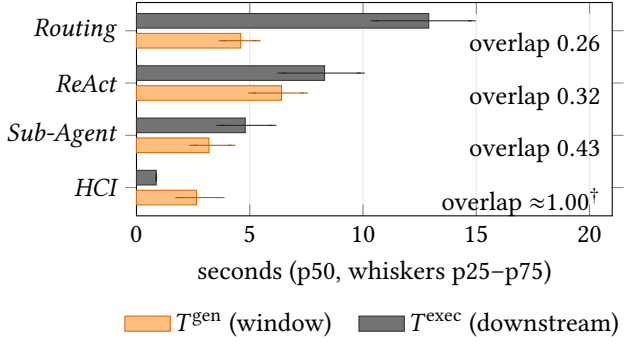

\vspace{0.05in}\noindent\textbf{Predictability determines coverage cost.} Early execution
saves time only if the admitted set contains the branch eventually
selected, and predictability differs by class. A bounded space is
enumerable from the template, and a skewed mix lets a few candidates
cover most traffic: bounded API selection reaches $91$--$94\%$
with an oracle retriever and $59$--$84\%$ with a deployed
one~\cite{patil2024gorilla}; Figure~\ref{fig:funnel} spans both. A user's
next action follows a frequency prior over those offered
before~\cite{fitchett12accessrank,davison98commands}; in \wl{HCI} the top
candidate matches the selection in $29$ of $30$ turns. An open vocabulary
has neither~\cite{toolret}, so candidates come from context already read:
our \wl{ReAct} predictor ranks entities named in the previous hop's
passage, placing the next in its top three $61\%$ of the time. More
candidates improve coverage~\cite{dpr} at proportionate cost.

\vspace{0.05in}\noindent\textbf{Recurrence concentrates in content sub-tasks.} What a fill
cannot hide, reuse removes outright when the same work is asked for
twice. With speculation disabled, $12\%$ of \wl{ReAct} retrievals in a run and
$26\%$ of \wl{Sub-Agent} sub-tasks across runs have byte-identical
inputs, and recurrence increases as more requests arrive
(Figure~\ref{fig:recurrence}a): repeated \textsc{MuSiQue}~\cite{musique}
sub-topic mentions rise from $19\%$ at $100$ questions to $61\%$ at
$5{,}000$. It appears outside our workloads too: our analysis of released
SWE-agent and SWE-rebench trajectories~\cite{sweagent,swerebench} finds
$42$--$45\%$ of commands recurring verbatim within a repository.
A production prefix-cache study reports $54$--$62\%$ ideal
hits~\cite{kvcachewild}; that unit requires a shared prefix, while the retrievals and sub-tasks above recur
in requests that differ elsewhere: the unit C1 must name.

\vspace{0.05in}\noindent\textbf{Speculation's cost depends on load.} A fill's engine
time pays off only if its result is used. Serving leaves capacity for
it~\cite{mlaas,splitwise}: over a day, GPUs leave
$\approx$$81\%$ of streaming-multiprocessor (SM) time idle~\cite{rose},
decode leaves $>$$50\%$ of compute unused~\cite{memgap}, and human
thinking is $64\%$ of coding-session wall-clock with idle gaps capped at an
hour~\cite{tracelab}. Fills still compete with canonical requests
for decode slots and memory bandwidth; Figure~\ref{fig:cobatch}b
measures p50
latency for a fixed-shape 32B request as fills grow in number and
context. At a fixed context
length, $1$--$3$ co-resident fills raise it by $4$--$6\%$ and $48$ by
$53\%$; holding the count at $24$ and growing each context from none to
$4{,}500$ words raises it from $11\%$ to $145\%$, as the fills grow from
$3\%$ to $24\%$ of the KV pool. Resident KV rises with both, which is why C3's
admission reads it as the measure of speculation in flight
(\S\ref{sec:design:admission}).

\begin{figure}[t]
\centering
\begin{tikzpicture}
\begin{axis}[at={(0.125\columnwidth,0)}, anchor=south west,
  scale only axis, width=0.340\columnwidth, height=0.30\columnwidth,
  title={(a) recurrence},
  xmode=log, log ticks with fixed point,
  xtick={100,1000,5000},
  xticklabels={100,1k,5k},
  xlabel={stream length},
  ylabel={recurrence fraction},
  ymin=0, ymax=1.05,
  ymajorgrids, grid style={gray!20},
  tick label style={font=\normalsize}, label style={font=\normalsize},
  title style={font=\normalsize, yshift=9pt},
  legend to name=fig5recurrencelegend,
  legend style={legend columns=2, inner ysep=0pt,
    /tikz/nodes={inner ysep=0pt},
    font=\normalsize, draw=none, fill=none, /tikz/every even column/.append style={column sep=5pt}},
  clip=false,
]
\addplot[color=blue!60!black, mark=*, mark size=1.5pt, thick] coordinates
  {(100,0.23)(500,0.63)(1000,0.75)(2000,0.82)(5000,0.89)};
\addlegendentry{questions with a repeat}
\addplot[color=green!55!black, mark=square*, mark size=1.4pt, thick] coordinates
  {(100,0.19)(500,0.42)(1000,0.51)(2000,0.56)(5000,0.61)};
\addlegendentry{repeated mentions}
\end{axis}
\begin{axis}[at={(0.500\columnwidth,0)}, anchor=south west,
  scale only axis, width=0.340\columnwidth, height=0.30\columnwidth,
  xmode=log, log basis x=2,
  xtick={1,3,12,24,48}, xticklabels={1,3,12,24,48},
  xlabel={co-resident fills},
  ylabel={canonical stretch (\%)},
  yticklabel pos=right, ylabel near ticks,
  title={(b) interference},
  legend to name=fig5interferencelegend,
  legend style={legend columns=2, font=\normalsize, draw=none, fill=none,
    inner ysep=0pt, /tikz/nodes={inner ysep=0pt},
    /tikz/every even column/.append style={column sep=5pt}},
  ymode=log, log ticks with fixed point,
  ymin=2, ymax=200, ytick={3,10,30,100},
  ymajorgrids, grid style={gray!20},
  tick label style={font=\normalsize}, label style={font=\normalsize},
  title style={font=\normalsize, yshift=9pt},
  clip=false,
]
\addplot[color=red!70!black, mark=*, mark size=1.5pt, thick] coordinates
  {(1,4.3)(2,4.8)(3,5.9)(6,11.7)(9,14.8)(12,15.4)(24,30.1)(48,53.4)};
\addlegendentry{slots}
\addlegendimage{color=violet!65!black, mark=square*, mark size=1.4pt, thick, dashed}
\addlegendentry{context (words)}
\end{axis}
\begin{axis}[at={(0.500\columnwidth,0)}, anchor=south west,
  scale only axis, width=0.340\columnwidth, height=0.30\columnwidth,
  axis x line*=top,
  xmin=-350, xmax=4900,
  xtick={0,1500,4500}, xticklabels={0,1.5k,4.5k},
  xtick align=outside, xminorticks=false,
  axis line style={violet!65!black}, tick style={violet!65!black},
  ymode=log, ymin=2, ymax=200,
  ytick=\empty, axis y line=none,
  x label style={font=\normalsize},
  tick label style={font=\normalsize, text=violet!65!black},
  clip=false,
]
\addplot[color=violet!65!black, mark=square*, mark size=1.4pt, thick, dashed]
  coordinates {(0,10.8)(1500,50.5)(4500,144.5)};
\end{axis}
\coordinate (figfivelegendanchor) at (current bounding box.south);
\node[anchor=north, inner sep=0pt, outer sep=0pt] (figfivelegenda)
  at ([yshift=-3pt]figfivelegendanchor) {\pgfplotslegendfromname{fig5recurrencelegend}};
\node[anchor=north, inner sep=0pt, outer sep=0pt]
  at (figfivelegenda.south) {\pgfplotslegendfromname{fig5interferencelegend}};
\end{tikzpicture}
\caption{(a)~Share of \textsc{MuSiQue} questions that repeat an earlier sub-topic, and of
sub-topic mentions that are repeats, as the stream grows. (b)~Slowdown
of a canonical request as fills grow in number and context.}
\Description{Panel a plots, as a MuSiQue question stream grows, the share
of questions that repeat an earlier sub-topic and the share of sub-topic
mentions that are repeats. Panel b plots canonical decode slowdown as either the
number of co-resident fills or the context each one holds
increases.}
\label{fig:recurrence}
\label{fig:cobatch}
\end{figure}
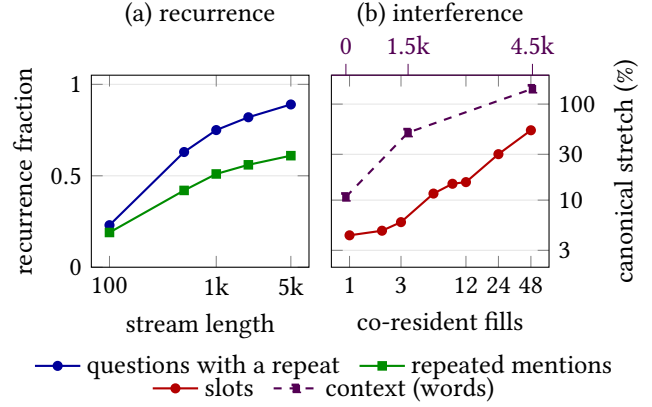

%% file: sections/design.tex
\section{\sys{} Design}
\label{sec:overview}
\label{sec:design:opt}
\label{sec:design}

\subsection{System Overview}

\sys{} is a serving layer between agentic harnesses and model-serving
engines, built from one interface and three parts. The \emph{candidate
interface} takes a branch coordinate and returns a ranked candidate
list; which predictor sits behind it is a deployment choice
(\S\ref{sec:design:lifecycle}). A two-level \emph{gate}
prices whether to open a branch, and which of its candidates to fill,
against current load. The \emph{Completed-Subgraph Library} (CSL)
holds finished subgraph results for later reuse, and a reverse index
from the external records they read invalidates them. A per-branch record
follows each open branch until it resolves and returns what it cost and
what it saved to the gate.

\begin{figure*}[t]
\centering
\resizebox{\textwidth}{!}{%
\begin{tikzpicture}[
  font=\normalsize,
  lane/.style={black!28, line width=0.4pt},
  lnm/.style={font=\normalsize\bfseries, anchor=east, inner sep=0pt},
  ttl/.style={font=\normalsize\bfseries, anchor=west, inner sep=0pt},
  ev/.style={draw=black, line width=0.5pt, fill=white, inner xsep=3pt, inner ysep=1.8pt,
             minimum height=11pt, anchor=west, align=center},
  hostev/.style={ev, fill=fmpale},
  reuse/.style={ev, draw=green!30!black, line width=0.9pt, fill=fmreuse},
  wrt/.style={ev, draw=red!70!black, text=red!70!black},
  cmsg/.style={-{Latex[length=4pt, width=3pt]}, black, line width=0.5pt},
  smsg/.style={-{Latex[length=4pt, width=3pt]}, orange!85!black, line width=0.5pt,
               densely dashed},
  gmsg/.style={-{Latex[length=4pt, width=3pt]}, green!30!black, line width=1.0pt},
  wmsg/.style={-{Latex[length=4pt, width=3pt]}, red!70!black, line width=0.5pt},
  cased/.style={preaction={draw=white, line width=2.2pt}},
  tag/.style={font=\normalsize, text=black!55, inner sep=1.5pt},
]
\def\yA{0}     \def\yH{-0.74} \def\yEc{-1.44} \def\yE{-1.87} \def\yF{-2.28}
\def\yS{-2.96} \def\bh{0.18}
\def\xao{1.90} \def\xai{12.25} \def\xbo{12.75} \def\xbi{17.60}
\foreach \y in {\yA,\yH,\yS} {
  \draw[lane] (\xao,\y) -- (\xai,\y);
  \draw[lane] (\xbo,\y) -- (\xbi,\y);
}
\fill[black!4] (\xao,-2.50) rectangle (\xai,-1.24);
\fill[black!4] (\xbo,-2.50) rectangle (\xbi,-1.24);
\node[lnm] at (1.72,\yA) {Agent};
\node[lnm] at (1.72,\yH) {\sys{}};
\node[lnm] at (1.72,\yE) {Engines};
\node[lnm] at (1.72,\yS) {Stores};
\draw[black!45, line width=0.4pt, densely dashed]
  (12.50,0.42) -- (12.50,-3.28);
\node[ttl] at (\xao,0.74) {(a)~one request: fill early, consume at the demand};
\node[ttl] at (\xbo,0.74) {(b)~later demands and writes};
\node[ev]     (A1) at (\xao,\yA) {request};
\node[hostev] (H1) at (2.40,\yH) {propose~\cnum{1}, price~\cnum{2}};
\filldraw[fill=fmpale, draw=black, line width=0.5pt]
  (2.95,\yEc-\bh) rectangle (7.75,\yEc+\bh);
\node at (5.95,\yEc) {canonical decode ($T^{\mathrm{gen}}$)};
\draw[cmsg] (2.60,0 |- A1.south) -- (2.60,0 |- H1.north);
\draw[cmsg] (3.05,0 |- H1.south) -- (3.05,\yEc+\bh);
\filldraw[fill=fmwin, draw=orange!85!black, line width=0.5pt]
  (3.55,\yE-\bh) rectangle (7.95,\yE+\bh);
\node at (5.75,\yE) {fill $h_1$~\cnum{3}};
\filldraw[fill=fmwin, draw=orange!85!black, line width=0.5pt]
  (3.55,\yF-\bh) rectangle (5.80,\yF+\bh);
\node at (4.68,\yF) {$h_2$};
\node[text=red!70!black, anchor=west, inner sep=2pt] at (5.83,\yF) {\ding{55}};
\node[tag, anchor=west] at (6.20,\yF) {discarded};
\draw[smsg, -] (3.45,0 |- H1.south) -- (3.45,\yF);
\draw[smsg] (3.45,\yE) -- (3.59,\yE);
\draw[smsg] (3.45,\yF) -- (3.59,\yF);
\node[ev, anchor=center] (S1) at (7.85,\yS) {shadow~\cnum{4}};
\draw[smsg] (7.85,\yE-\bh) -- (7.85,0 |- S1.north);
\node[ev] (A2) at (7.90,\yA) {stage request};
\draw[{Latex[length=4pt, width=3pt]}-{Latex[length=4pt, width=3pt]}, black!45,
      line width=0.4pt] (A1.east) -- (A2.west);
\node[tag, anchor=south, fill=white, inner sep=1.5pt] at (5.50,-0.16) {resolution window};
\draw[cmsg, cased] (7.75,\yEc) -- (8.05,\yEc) -- (8.05,0 |- A2.south);
\fill[black] (8.05,\yH) circle (1.1pt);
\node[hostev] (H2) at (9.20,\yH) {authorize~\cnum{5}};
\draw[cmsg] (9.50,0 |- A2.south) -- (9.50,0 |- H2.north);
\node[reuse] (S2) at (9.55,\yS) {promote $h_1$};
\draw[gmsg] (10.90,0 |- H2.south) -- (10.90,0 |- S2.north);
\node[ev] (A3) at (10.85,\yA) {response};
\draw[gmsg, cased] (11.30,0 |- S2.north) -- (11.30,0 |- A3.south);
\draw[black!45, line width=0.4pt, densely dashed]
  (8.20,-2.12) rectangle (10.75,\yEc+\bh);
\node[tag, align=center] at (9.475,-1.69) {saved\\serialized work};
\node[ev] (A4) at (\xbo+0.10,\yA) {later demand};
\node[reuse] (S3) at (12.85,\yS) {hit~\cnum{6}};
\draw[gmsg, cased] (13.00,0 |- A4.south) -- (13.00,0 |- S3.north);
\draw[gmsg, cased] (13.70,0 |- S3.north) -- (13.70,0 |- A4.south);
\node[tag, anchor=west] at (14.30,\yEc) {no engine work};
\node[wrt, anchor=west] (S4) at (14.45,\yS) {write: invalidate~\cnum{7}};
\node[wrt, anchor=west] (E4) at (15.55,\yF) {evict KV};
\draw[wmsg] (16.05,0 |- S4.north) -- (16.05,0 |- E4.south);
\draw[lane] (\xao,-3.38) -- (\xbi,-3.38);
\node[tag, anchor=west] at (\xao,-3.56)
  {\textbf{C1} \S\ref{sec:design:smg} exact reuse: \cnum{6}~\cnum{7}};
\node[tag, anchor=west] at (7.20,-3.56)
  {\textbf{C2} \S\ref{sec:design:lifecycle} early fill: \cnum{3}~\cnum{4}~\cnum{5}};
\node[tag, anchor=west] at (12.75,-3.56)
  {\textbf{C3} \S\ref{sec:design:admission} admission: \cnum{1}~\cnum{2}~\cnum{3}};
\end{tikzpicture}}
\caption{Lifecycle of a speculative branch result. Lanes are components,
time runs left to right; \ding{55} marks the discarded loser.}
\Description{A four-lane timeline with agent, \sys{}, engines, and stores
lanes and time running left to right. In the first panel a request opens a
branch; \sys{} proposes and prices candidates, whose admitted fills execute
on the engines concurrently with the canonical decode that names the stage,
for the whole window in which the agent is blocked. The completed fill becomes an invisible
shadow in the stores and the losing candidate is reclaimed. When the decode
resolves, the selection returns to the agent, whose stage request authorizes
the matching shadow; it is promoted and consumed to answer without further
engine work, marked by a dashed slot labeled saved serialized work. In the second
panel a later exact demand is served from the store while the engines stay
idle, and an external write invalidates both the stored entry and
engine-resident state.}
\label{fig:overview}
\end{figure*}
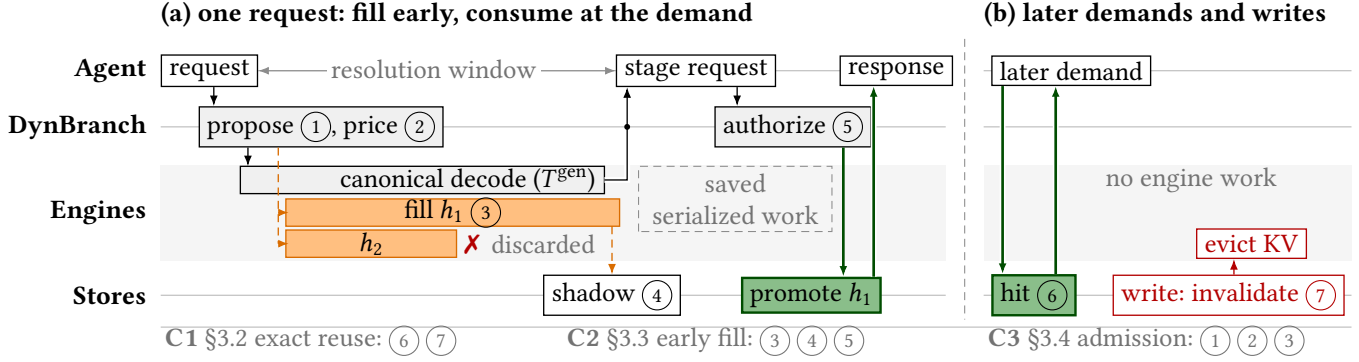

\vspace{0.05in}\noindent\textbf{Life of a request.}
Figure~\ref{fig:overview} follows a model-decided branch: its resolution
window lets predicted downstream work run before the next stage is named.
The predictor uses the branch \emph{coordinate}, the stable
identity of that decision in the workflow template, to propose
downstream candidates (\cnum{1}). The gate prices each candidate and admits those whose
expected benefit exceeds its load price (\cnum{2}); the admitted fills then run
on the same engines as canonical traffic (\cnum{3}), concurrently
with the branch-resolution decode.

A completed fill remains an invisible \emph{shadow}, a speculative
result that serving paths cannot read (\cnum{4}). Branch resolution eliminates
losers but does not authorize the survivor: only a subsequent
canonical demand with the same exact reuse name can promote and
consume it (\cnum{5}). This separation takes the serialized work off the
critical path without exposing unchosen work.
The confirmed entry then serves later
exact matches (\cnum{6}) until a dependency
write invalidates it (\cnum{7}).

\sys{} treats reuse, overlap, and execution as separate
decisions.
\emph{C1: what may be reused?} A subgraph, identified by its exact
realized input and fresh dependencies (\S\ref{sec:design:smg}).
\emph{C2: how can work run before resolution?} A stable coordinate starts
fills during the window; a matching canonical demand authorizes use
(\S\ref{sec:design:lifecycle}).
\emph{C3: when should that work run?} When its expected latency
benefit exceeds its load-priced resource cost
(\S\ref{sec:design:admission}).

\subsection{Shared-Subgraph Reuse (C1)}
\label{sec:design:smg}

C1 must capture reuse that exists below the whole-request boundary or
across requests that differ elsewhere, without serving a hit it cannot
replay. Replay safety asks more than a semantic
cache's approximate match~\cite{gptsemcache}: every result-determining
input must match, every external dependency must remain fresh, and a
multi-node result must come from one execution, not fragments of several.

\vspace{0.05in}\noindent\textbf{Subgraph reuse names.}
Exact reuse follows the dependency closure captured by a rendered input,
not a node's graph position. A whole-response key needs the whole workflow
to repeat~\cite{gptcache} and a prefix-region key an unchanged
prefix~\cite{sglang,kvflow}. \sys{} instead makes the unit a
self-contained subgraph, depending only on the upstream values entering
its own rendered inputs, so it stays valid beyond an unrelated change.
It records that subgraph with the realized input, output, external
read-set, and provenance of one execution. Its \emph{reuse name} combines a subgraph fingerprint
with a hash of the realized input, covering the request fields that
determine the output. C2
addresses work before resolution through the branch coordinate, which
identifies an unresolved decision and therefore cannot authorize reuse.

\label{sec:design:cascade}\label{sec:impl:cascade}
Freshness is maintained by a reverse index from external records to
cached entries. A write invalidates direct dependents and any cached
or engine-resident state that loses a valid producer; incomplete
provenance conservatively extends the cascade. Reads synchronize
with invalidation, so concurrent or uncertain cases fall back to a
miss. Writes may reduce reuse but cannot expose stale state; the cascade
follows view-maintenance and lineage
work~\cite{blakeley1986,noria,sac}.
Together, these checks make a hit replay one recorded execution~\cite{rr}.

\vspace{0.05in}\noindent\textbf{Benefit of reuse.}
Let $T_j^{\mathrm{sub}}$ be the
execution the subgraph rooted at decision point $j$ adds beyond its nested
subgraphs, and $P_j^{\mathrm{hit}}$ the probability that a later request is
served it from cache while its recorded reads stay fresh. In expectation,
reuse removes
\begin{equation}
\label{eq:reuse}
  \mathbb{E}[\Delta T_{\mathrm{reuse}}]
  = \sum_j P_j^{\mathrm{hit}}\,T_j^{\mathrm{sub}}
\end{equation}
from that request's critical path. The boundary sets $P^{\mathrm{hit}}$:
under one realized change
(Figure~\ref{fig:ladder}), the whole-response, prefix-region, and
subgraph boundaries serve $0$, $2$, and $4$ of six taken nodes at the same
exactness and freshness. A cache that does not model what the stage
computes cannot match more often than this and stay exact, so promotion
refuses nothing a safe rule could admit
(\hyperlink{supp-c1}{Corollary~A.4}).
C3 values only the execution
Equation~\eqref{eq:reuse} leaves.

\begin{figure}[t]
\centering
\begin{tikzpicture}[
  font=\normalsize,
  nd/.style={draw=black, line width=0.8pt, circle, minimum size=14pt, inner sep=1pt, fill=white},
  hit/.style={nd, fill=fmreuse, draw=green!30!black},
  chg/.style={nd, draw=red!80!black},
  cold/.style={nd},
  unt/.style={nd, densely dotted, draw=black!55, text=black!55},
  xm/.style={font=\normalsize\bfseries, text=red!80!black, inner sep=0pt,
    anchor=south west},
  fl/.style={-{Latex[length=3.4pt, width=2.4pt]}, line width=0.35pt, shorten >=1pt, shorten <=1pt},
  ctl/.style={fl, densely dashed, black!55},
  ttl/.style={font=\normalsize\bfseries, align=center},
  num/.style={font=\normalsize\bfseries, text=black!80, anchor=base},
  tag/.style={font=\normalsize, text=black!55, anchor=base},
]
\begin{scope}
\node[ttl] at (0.62, 4.45) {whole response};
\node[cold] (a1) at (0, 3.9) {$v_1$};
\node[cold] (a2) at (0, 3.2) {$v_2$};
\node[chg]  (a3) at (0, 2.5) {$v_3$};
\node[xm] at ([shift={(1.5pt,1.5pt)}]a3.north east) {\ding{55}};
\node[cold] (a4) at (0, 1.8) {$v_4$};
\node[unt]  (au) at (1.25, 3.2) {$u$};
\node[cold] (as1) at (1.25, 2.05) {$x_1$};
\node[cold] (as2) at (1.25, 1.3) {$x_2$};
\draw[fl] (a1)--(a2); \draw[fl] (a2)--(a3); \draw[fl] (a3)--(a4);
\draw[fl, densely dotted, black!55] (a2)--(au);
\draw[ctl] (a3)--(as1); \draw[ctl] (a4)--(as2);
\node[num] at (0.62, 0.62) {$\mathbf{0/6}$};
\node[tag] at (0.62, 0.15) {served};
\end{scope}
\begin{scope}[xshift=2.7cm]
\node[ttl] at (0.62, 4.45) {prefix region};
\node[hit]  (b1) at (0, 3.9) {$v_1$};
\node[hit]  (b2) at (0, 3.2) {$v_2$};
\node[chg]  (b3) at (0, 2.5) {$v_3$};
\node[xm] at ([shift={(1.5pt,1.5pt)}]b3.north east) {\ding{55}};
\node[cold] (b4) at (0, 1.8) {$v_4$};
\node[unt]  (bu) at (1.25, 3.2) {$u$};
\node[cold] (bs1) at (1.25, 2.05) {$x_1$};
\node[cold] (bs2) at (1.25, 1.3) {$x_2$};
\draw[fl] (b1)--(b2); \draw[fl] (b2)--(b3); \draw[fl] (b3)--(b4);
\draw[fl, densely dotted, black!55] (b2)--(bu);
\draw[ctl] (b3)--(bs1); \draw[ctl] (b4)--(bs2);
\node[num] at (0.62, 0.62) {$\mathbf{2/6}$};
\node[tag] at (0.62, 0.15) {unchanged prefix};
\end{scope}
\begin{scope}[xshift=5.4cm]
\node[ttl] at (0.62, 4.45) {$+$ subgraph splice};
\node[hit]  (c1) at (0, 3.9) {$v_1$};
\node[hit]  (c2) at (0, 3.2) {$v_2$};
\node[chg]  (c3) at (0, 2.5) {$v_3$};
\node[xm] at ([shift={(1.5pt,1.5pt)}]c3.north east) {\ding{55}};
\node[cold] (c4) at (0, 1.8) {$v_4$};
\node[unt]  (cu) at (1.25, 3.2) {$u$};
\node[hit] (cs1) at (1.25, 2.05) {$x_1$};
\node[hit] (cs2) at (1.25, 1.3) {$x_2$};
\draw[fl] (c1)--(c2); \draw[fl] (c2)--(c3); \draw[fl] (c3)--(c4);
\draw[fl, densely dotted, black!55] (c2)--(cu);
\draw[ctl] (c3)--(cs1); \draw[ctl] (c4)--(cs2);
\node[num] at (0.62, 0.62) {$\mathbf{4/6}$};
\node[tag] at (0.62, 0.15) {$+$ self-contained};
\end{scope}
\node[num, text=black!60] at (1.97, 0.62) {$<$};
\node[num, text=black!60] at (4.67, 0.62) {$<$};
\end{tikzpicture}
\caption{One realized change (red \ding{55}) under three exact serving
boundaries. Green: served from cache; dashed arrows reach control-only
descendants; dotted $u$ is an untaken branch. Footers: nodes served of
six taken.}
\Description{Three copies of the same workflow show whole-response,
prefix-region, and subgraph-splice reuse after one node changes.
Whole-response reuse serves nothing, region reuse serves the
unchanged prefix, and subgraph splicing additionally serves two
self-contained downstream tasks whose inputs exclude the changed
value.}
\label{fig:ladder}
\end{figure}
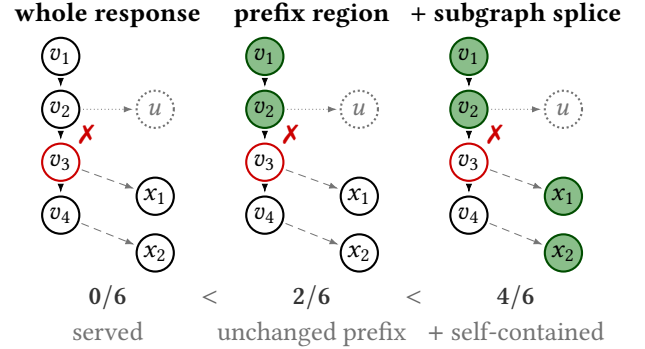

On a canonical demand, C1 derives the reuse name and probes the CSL's
confirmed entries under the freshness guard. A hit returns the recorded
result; a miss executes normally and may refill the CSL, which is bounded
and evicts by recency, frequency, and cost per byte.

\subsection{Early Fill During Resolution (C2)}
\label{sec:design:lifecycle}

C2 must avoid two failures. The first is serial: downstream work that
could have run inside the resolution window waits for it instead, and the
obvious repair (publishing a result as soon as it completes) serves work
no demand selected~\cite{smith1981}. The second is cross-request:
speculative decoding, speculative agents, and speculative tool systems bind
validation and consumption to the request that authorized the
work~\cite{leviathan2023,specactions,dsap,paste,spectool}, so a result that
survives its branch still dies with its request. C2 therefore decouples
computation, selection, and authorization.

\vspace{0.05in}\noindent\textbf{Where candidates come from.}
Any predictor can sit behind the interface: a next-stage model over the
coordinate's history, a retrieval ranker, or a small draft LM sampled $N$
times and ranked by how often it names a candidate
(\S\ref{sec:eval:cost}). Only its ordering is used, since its scores are
uncalibrated; a per-rank consumption rate, learned online, replaces them.

\vspace{0.05in}\noindent\textbf{Speculative results and promotion.}
C2 speculates only on \emph{speculable} stages, those the template declares
safe to run twice (\S\ref{sec:bg:wall-clock}); unannotated, sampling, or
externally mutating stages run only after resolution.
Admission turns a candidate into a \emph{fill}. Once the fill's input
renders, C1 derives the reuse name, and completion records that name, the
result, and its read-set as a persistent shadow,
visible to the gate's per-branch record but unreadable by serving paths.

A canonical demand resolves this lifecycle in two steps. First,
selection settles the branch. For an exclusive branch, the resolved choice
retains the matching candidate and aborts the losers' in-flight fills in
the engine, releasing their slots and KV and discarding their results;
independent offers remain eligible, their shadows
unreadable until an exact demand authorizes them or they expire. An
ambiguous branch or choice promotes nothing. Second, authorization applies
C1's guards: the canonical reuse name must match the shadow's,
and its recorded read-set must remain fresh. Only then is the shadow
atomically promoted and consumed: selection
eliminates work, content and freshness serve it.

If the matching fill is still running, the canonical demand waits a
bounded interval and otherwise executes normally. Multi-stage
fills are promoted stage by stage as canonical demands pass the same
guards.

\vspace{0.05in}\noindent\textbf{Speculative--canonical equivalence.}
Let $F(q,\theta,D)$ be a speculable stage run with rendered request
$q$, model and execution parameters $\theta$, and dependency versions
$D$, and write $s$ and $c$ for its speculative and canonical executions.
Promotion requires a reuse-name match under the freshness guard. The name
covers the stage fingerprint, the complete rendered input, and both
parameter sets, so a match gives $(q_s,\theta_s)=(q_c,\theta_c)$, and the
guard gives $D_s=D_c$. Determinism then yields
\begin{equation*}
  y_s=F(q_s,\theta_s,D_s)=F(q_c,\theta_c,D_c)=y_c ,
\end{equation*}
and a promoted stage $i$ renders to stage $i+1$ the input its canonical
execution would have rendered, re-establishing the hypothesis along a
multi-stage fill. Two properties follow. \emph{Value equivalence}: a
promoted subgraph returns exactly what canonical execution would.
\emph{Non-interference}: unpromoted work is never observed---no external
effect, no readable state---and exclusive losers do not outlive their branch. A failed precondition is a
miss, so speculation changes only latency.

\vspace{0.05in}\noindent\textbf{Benefit of early fill.}
At branch $k$, let $T_k^{\mathrm{gen}}$ be the generation
window that resolves the branch, $T_k^{\mathrm{exec}}$ the selected
downstream work, and $P_k^{\mathrm{cov}}$ the probability that
the admitted set contains the selected candidate. Ignoring contention,
early fill at branch $k$ hides at most
$\min(T_k^{\mathrm{gen}},T_k^{\mathrm{exec}})$ if the admitted set covers
the selection and nothing otherwise (Figure~\ref{fig:barrier}). With the
durations fixed, in expectation over coverage:
\begin{equation}
\label{eq:saving}
  \mathbb{E}[\Delta T_{\mathrm{gross}}]
  \le \sum_k P_k^{\mathrm{cov}}\,
    \min\!\bigl(T_k^{\mathrm{gen}},T_k^{\mathrm{exec}}\bigr).
\end{equation}
Equation~\eqref{eq:saving} is a gross bound
(\hyperlink{supp-overlap}{Lemma~A.6}): C3 prices contention
explicitly~\cite{specqueue,vulimiri,synergy}, and C2 admits fills only while a
positive overlap window remains. The opportunity grows with branch
predictability, downstream cost, and window length.
For a user-wait window, the same expression substitutes think time for
generation time and measures the reduction in post-click
latency~\cite{doherty,brutlag2009}.

\subsection{Load-Priced Admission (C3)}
\label{sec:design:admission}

\sys{} admits speculative work in two steps (Algorithm~\ref{alg:admission}).
For each branch, the \emph{generation gate} (Action~0) decides whether
running the predictor and expanding candidates is worth its fixed cost.
If it opens, the \emph{fill gate} (Action~1) admits a candidate only if
the expected time it saves exceeds the expected delay it imposes on
canonical requests at current load, re-pricing after each admission. The
two are not one decision: generation itself costs time that no fill
threshold can recover, and a cheap generation step can still reveal fills
that are individually too costly.

Predicting more actions raises the chance of covering the realized branch, but
also creates work that may never be
used~\cite{specactions,spagent,specinfer,multicand}. Admission must
therefore protect canonical demand from those fills.
Paying for an alternative before knowing whether it is needed is classical
optimal search~\cite{weitzman1979}; here the alternative is produced, not
fetched, so its price moves with contention.

\vspace{0.05in}\noindent\textbf{Causal load price.}
At time $t$, C3 observes a load view $V_t$ comprising normalized
fill occupancy $o_{\mathrm{fill}}$, reservation-budget occupancy
$o_{\mathrm{KV}}$, and canonical-demand pressure $\rho_t$. Admission is the
only point at which contention can be refused: a running fill holds its
slots until it completes or its branch cancels it, so an instantaneous
idle reading would underprice a multi-second fill. C3 prices against the
causal envelope~\cite{fdp}
\begin{equation}
\label{eq:horizon-pressure}
 \bar\rho_H(t)=\max\!\left\{\rho_t,
   \frac{1}{H}\int_{t-H}^{t}\rho_\tau\,d\tau\right\}.
\end{equation}
Here $H$ matches the window used to estimate speculative co-residency.
Rising
load takes effect immediately; falling load relaxes the price gradually
over a full execution window. A periodic probe updates this history even
while the generation gate is closed.

Candidate $s$ carries footprint
$\hat{\mathbf R}_s=(\hat R_{\mathrm{slot},s},\hat R_{\mathrm{KV},s})$,
covering decode time and KV residency weighted by its budget share. Monotone
price curves
$\boldsymbol\beta(V_H)=(\beta_{\mathrm{slot}},\beta_{\mathrm{KV}})$, in
milliseconds of canonical delay per unit of each resource, convert the
footprint into a delay cost,
$C_s(V_H)=\boldsymbol\beta(V_H)^{\!\top}\hat{\mathbf R}_s$, where $V_H$
is the load view $V_t$ with $\bar\rho_H$ in place of $\rho_t$. The slot
price scales with $\bar\rho_H$ and the KV price with budget occupancy; both
are initialized from Figure~\ref{fig:cobatch}b and steepen near capacity.
Measured prices govern value; fill-count and atomic KV reservations remain
hard feasibility constraints.

\vspace{0.05in}\noindent\textbf{Generation gate (Action 0).}
A branch's payoff is known only after it settles, so C3 learns one
full-path return per branch entry:
\begin{equation}
\label{eq:plane-target}
 \hat y_b=\sum_{s\in\mathrm{consume}(b)}\hat G_s^{\Delta}
     -T_{\mathrm{setup},b}-\hat K_b-\hat W_b .
\end{equation}
Here $\mathrm{consume}(b)$ is the set of fills branch $b$ consumed and
$\hat G_s^{\Delta}$ is the execution a consumed fill saves after C1 lookup,
so the first term credits only the residual C1 left. The remaining terms charge measured prediction and
expansion time $T_{\mathrm{setup},b}$, the cost $\hat K_b$ of the fills this
branch enabled, and $\hat W_b$ for expired or uncovered waits. Using one
aggregate target is essential: summing separate optimism bonuses over
mutually exclusive candidates would overvalue a branch.

A branch class $c$ groups the branches that share a coordinate, so repeat
visits to one decision point pool their statistics. For each class,
context $x$ summarizes horizon pressure, whether a
low reading is transient, recent fill consumption, and queue backlog.
A contextual ridge model maintains
$A_c=I+\sum_{b\in c}x_bx_b^\top$ and
$q_c=w_{0,c}+\sum_{b\in c}x_b\hat y_b$ from a class prior $w_{0,c}$. Following
linear UCB~\cite{li2010linucb}, the gate opens when
\begin{equation}
\label{eq:enabling}
 \hat U_0(c,x)=x^\top A_c^{-1}q_c+
   \hat\sigma_c\sqrt{x^\top A_c^{-1}x}>0 ,
\end{equation}
where $\hat\sigma_c$ scales exploration by the observed target RMS. C3
allows $n_0$ cold trials per class, counted by $n_c$, before applying it,
and registers
each trial in the branch's record before expansion, so feedback arriving
during setup is not lost (\hyperlink{supp-action0}{App.~A, generation gate}).

A predictor too expensive to run inline is time-boxed off the dispatch
path and yields an empty set at its deadline, which bounds
$T_{\mathrm{setup},b}$. The time it spent is charged either way, so a branch class whose
predictor repeatedly pays setup and returns nothing consumable drives
$\hat y_b$ negative and closes its own gate.

\vspace{0.05in}\noindent\textbf{Fill gate (Action 1).}
For candidate $s$, the optimistic probability of useful service factors as
$\tilde p_s=\tilde\pi_s^{\mathrm{sel}}\bar\pi_s^{\mathrm{use}}$: a
rank-conditioned Beta upper quantile~\cite{bayesucb} times the posterior mean
of timely usability conditional on selection. It clears the priced rule
when
\begin{equation}
\label{eq:ev}
 U_s(V_H)=\tilde p_s\hat G_s^{\Delta}
   -(1-\tilde p_s)C_s(V_H)>0 .
\end{equation}
Both terms are measured in milliseconds. If
$\hat G_s^{\Delta}+C_s(V_H)>0$, the rule is equivalent to
\begin{equation}
\label{eq:breakeven}
 \tilde p_s>p^\ast(s,V_H)=
 \frac{C_s(V_H)}{\hat G_s^{\Delta}+C_s(V_H)} .
\end{equation}
The threshold is exact at the current
projected load, so the same candidate can pass under headroom and fail under
contention. On the exploitation path, Action~1 admits the highest-utility
candidate that passes feasibility and atomic reservation, projects its
pressure, and reprices the rest~\cite{tip}. It stops when no remaining
candidate has positive utility (\hyperlink{supp-action1}{Proposition~A.8}). Seeded audits occasionally
admit refused work, so closed gates keep learning~\cite{vernade2020delayed}.

\begin{algorithm}[t]
\caption{Load-priced admission with delayed feedback.}
\label{alg:admission}
\begin{algorithmic}[1]
\Require load history, posteriors, branch-class models
\Statex \textbf{Action 0: generation gate for branch $b$}
\State $c\gets\Call{BranchClass}{b}$;
       $(x,V_H)\gets\Call{HorizonContext}{}$
\State \textbf{if} $n_c<n_0\le n_c+n_c^{\mathrm{pending}}$
       \textbf{then return} \Comment{cold trials in flight}
\State \textbf{if} $n_c\ge n_0\land\hat U_0(c,x)\le0\land\lnot$%
       \Call{Audit}{b,0} \textbf{then return}
\State $\tau\gets\Call{BeginTrial}{b,c,x}$
\State $\mathcal A\gets\Call{EligibleCandidates}{b}$;
       $\mathcal S\gets\emptyset$; $V'\gets V_H$
\Statex \textbf{Action 1: fill gate at projected load}
\While{$\mathcal A\ne\emptyset$}
  \State $s\gets\arg\max_{a\in\mathcal A}U_a(V')$;
         $\mathcal A\gets\mathcal A\setminus\{s\}$;
         $probe\gets0$
  \If{$U_s(V')\le0$}
    \State \textbf{if} $\lnot$\Call{Audit}{b,1} \textbf{then break};
           $probe\gets1$
  \EndIf
  \If{$\Call{HardFeasible}{V',s}\land\Call{Reserve}{b,s}$}
    \State record $(1-\tilde p_s)C_s(V')$ and the producer identity
    \State $\mathcal S\gets\mathcal S\cup\{s\}$;
           $V'\gets\Call{ProjectLoad}{V',s}$
  \EndIf
  \State \textbf{if} $probe=1$ \textbf{then break}
\EndWhile
\State \Call{SubmitFills}{$b,\mathcal S$};
       \Call{ChargeSetup}{$\tau,T_{\mathrm{setup},b}$}
\Statex \textbf{Feedback: event $e$ for trial $\tau$}
\State update selection, usability, gain, cost from $e$
\State \Call{ReviseTrial}{$\tau,\Delta_e\hat y_b$};
       \Call{ReleaseSettledResources}{$e$}
\end{algorithmic}
\end{algorithm}

When the branch resolves, C3 updates selection for every predicted
candidate, admitted or not; it updates usability only for admitted,
selected fills, so one that misses
its deadline for use after resolution lowers usability without becoming a prediction error. Delayed events revise the originating branch's
target rather than adding samples. Credit requires the consumed entry to
name the speculative producer recorded at admission (\S\ref{sec:eval:c1}).
Pricing classes separate single- and multi-stage fills, so a chain's
cost is compared with its full saving.

The three contracts now form one lifecycle: C3 admits fills, C2 runs
them during resolution while hiding their results, and exact, fresh
C1 demands promote and reuse them.
A C1 hit then suppresses the execution it answers and the speculation
below it, unless an external wait follows the stage.

%% file: sections/implementation.tex
\section{Implementation}
\label{sec:impl}

\noindent\textbf{Agentic harness integration.}
Framework-driven harnesses such as AutoGen~\cite{autogen},
LangGraph~\cite{langgraph}, and AgentScope~\cite{agentscope} attach to
\sys{} through an HTTP interceptor that tags each model request with
\texttt{workflow\_\allowbreak type\_\allowbreak id}, \texttt{workflow\_id}, and
\texttt{agent\_id}. Text-based coding-agent harnesses such as Claude
Code~\cite{claudecode}, Codex~\cite{codexcli}, and DeepSeek
Harness~\cite{dsh} normalize into the same arrival representation from
their own wire protocols, so \sys{} depends on neither the agent
loop nor the protocol.

\vspace{0.05in}\noindent\textbf{Serving integration.}
\sys{} hooks three lifecycle events on every normalized arrival:
stage arrival, the opening of an unresolved branch, and the canonical
request carrying its resolved selection. The hooks reuse the serving path's
existing queue, slot, and KV accounting, the same CSL manager, and the same
terminal settlement path; canonical requests and speculative fills execute
through the existing SGLang or vLLM executors~\cite{sglang,vllm}.
Disabling \sys{} bypasses the hooks and preserves the original semantics.

\begin{figure*}[!t]
\centering
\pgfplotsset{
  mutbar/.style={
    ybar, bar width=5pt, area legend,
    scale only axis, width=0.4230\textwidth, height=0.145\textwidth,
    ymin=0, ymax=1.12, ytick={0,0.5,1.0},
    ylabel={latency / floor}, ylabel style={font=\normalsize},
    symbolic x coords={routing,react,orch,menu},
    xtick={routing,react,orch,menu},
    xticklabels={\wl{Routing},\wl{ReAct},\wl{Sub-Agent},\wl{HCI}},
    tick label style={font=\normalsize}, enlarge x limits=0.18,
    xmajorgrids, ymajorgrids, grid style={black!18, line width=0.3pt},
    axis lines=box, axis line style={black!55, line width=0.25pt},
    xtick pos=bottom,
    tick style={black!65, line width=0.35pt}, tick align=outside,
    legend style={at={(0.999,-0.30)}, anchor=north, legend columns=7,
      font=\normalsize, draw=none, fill=none, /tikz/every even column/.append style={column sep=3pt}},
  },
  oursbar/.style={fill=blue!70!black, draw=blue!80!black},
  c1bar/.style={fill=blue!30, draw=blue!55!black},
  dspbar/.style={pattern=dots, pattern color=orange!85!black, draw=orange!70},
  spabar/.style={pattern=crosshatch, pattern color=orange!75!black, draw=orange!70},
  instbar/.style={fill=fmexec, draw=black!75},
  helibar/.style={pattern=north east lines, pattern color=teal!70!black, draw=teal!60},
  parrotbar/.style={pattern=horizontal lines, pattern color=black!75, draw=black!55},
}
\begin{tikzpicture}
\begin{axis}[mutbar, at={(0.049\textwidth,0.197\textwidth)}, anchor=south west,
  title={(a) low load \quad \textnormal{\normalsize 32B, 4$\times$H200}},
  title style={font=\normalsize\bfseries}, ytick pos=left, ymax=1.3,
  xticklabels=\empty, xlabel={}]
\addplot[oursbar] coordinates {(routing,0.456)(react,0.361)(orch,0.344)(menu,0.534)};
\addplot[c1bar] coordinates {(routing,0.523)(react,0.430)(orch,0.479)(menu,0.838)};
\addplot[dspbar] coordinates {(routing,0.942)(react,0.984)(orch,0.936)(menu,1.005)};
\addplot[spabar] coordinates {(routing,0.940)(react,0.993)(orch,0.946)(menu,1.017)};
\addplot[instbar] coordinates {(routing,0.739)(react,0.716)(orch,0.546)(menu,0.78)};
\addplot[helibar] coordinates {(routing,0.538)(react,0.419)(orch,0.486)(menu,0.816)};
\addplot[parrotbar] coordinates {(routing,1.049)(react,1.172)(orch,1.208)(menu,0.87)};
\addplot[dashed, gray!70, mark=none, sharp plot, forget plot] coordinates {(routing,1)(menu,1)};
\end{axis}
\begin{axis}[mutbar, at={(0.5047\textwidth,0.197\textwidth)}, anchor=south west,
  title={(b) at load \quad \textnormal{\normalsize 32B, 4$\times$H200}},
  title style={font=\normalsize\bfseries}, ylabel={}, ytick pos=right,
  yticklabel pos=right, ymax=1.3, xticklabels=\empty, xlabel={}]
\addplot[oursbar] coordinates {(routing,0.509)(react,0.487)(orch,0.363)(menu,0.540)};
\addplot[c1bar] coordinates {(routing,0.557)(react,0.489)(orch,0.516)(menu,0.828)};
\addplot[dspbar] coordinates {(routing,1.004)(react,1.012)(orch,0.971)(menu,1.017)};
\addplot[spabar] coordinates {(routing,0.991)(react,1.003)(orch,0.973)(menu,1.016)};
\addplot[instbar] coordinates {(routing,0.808)(react,0.826)(orch,0.579)(menu,0.783)};
\addplot[helibar] coordinates {(routing,0.596)(react,0.519)(orch,0.518)(menu,0.826)};
\addplot[parrotbar] coordinates {(routing,0.976)(react,1.067)(orch,1.244)(menu,1.044)};
\addplot[dashed, gray!70, mark=none, sharp plot, forget plot] coordinates {(routing,1)(menu,1)};
\end{axis}
\begin{axis}[mutbar, at={(0.049\textwidth,0)}, anchor=south west,
  title={(c) low load \quad \textnormal{\normalsize 8B, 4$\times$RTX4090}},
  title style={font=\normalsize\bfseries}, ytick pos=left, ymax=1.58,
  ytick={0,0.5,1.0,1.5}]
\addplot[oursbar] coordinates {(routing,0.634)(react,0.704)(orch,0.689)(menu,0.516)};
\addlegendentry{\sys}
\addplot[c1bar] coordinates {(routing,0.792)(react,0.840)(orch,0.831)(menu,0.842)};
\addlegendentry{C1-only}
\addplot[dspbar] coordinates {(routing,0.813)(react,0.938)(orch,0.799)(menu,1.018)};
\addlegendentry{DSP}
\addplot[spabar] coordinates {(routing,0.922)(react,1.018)(orch,0.856)(menu,0.993)};
\addlegendentry{SPAgent}
\addplot[instbar] coordinates {(routing,0.787)(react,0.849)(orch,0.819)(menu,0.827)};
\addlegendentry{InstCache}
\addplot[helibar] coordinates {(routing,0.797)(react,0.760)(orch,0.818)(menu,0.858)};
\addlegendentry{Helium$^{+}$}
\addplot[parrotbar] coordinates {(routing,1.150)(react,1.463)(orch,1.010)(menu,0.865)};
\addlegendentry{Parrot}
\addplot[dashed, gray!70, mark=none, sharp plot, forget plot] coordinates {(routing,1)(menu,1)};
\end{axis}
\begin{axis}[mutbar, at={(0.5047\textwidth,0)}, anchor=south west,
  title={(d) at load \quad \textnormal{\normalsize 8B, 4$\times$RTX4090}},
  title style={font=\normalsize\bfseries}, ylabel={}, ytick pos=right,
  yticklabel pos=right, ymax=1.58, ytick={0,0.5,1.0,1.5}]
\addplot[oursbar] coordinates {(routing,0.472)(react,0.454)(orch,0.688)(menu,0.544)};
\addplot[c1bar] coordinates {(routing,0.540)(react,0.565)(orch,0.760)(menu,0.856)};
\addplot[dspbar] coordinates {(routing,0.940)(react,1.134)(orch,0.932)(menu,1.039)};
\addplot[spabar] coordinates {(routing,0.896)(react,0.984)(orch,0.864)(menu,1.043)};
\addplot[instbar] coordinates {(routing,0.740)(react,0.720)(orch,0.793)(menu,0.862)};
\addplot[helibar] coordinates {(routing,0.530)(react,0.567)(orch,0.739)(menu,0.849)};
\addplot[parrotbar] coordinates {(routing,1.119)(react,1.270)(orch,1.042)(menu,0.864)};
\addplot[dashed, gray!70, mark=none, sharp plot, forget plot] coordinates {(routing,1)(menu,1)};
\end{axis}
\end{tikzpicture}
\caption{Open-loop mean latency, normalized to each workload's own
no-reuse floor (dashed $=1.0$). Rows are hardware; columns are load, at
$\lambda{=}0.12$ (\wl{HCI} $0.3$) and $0.8$\,req/s on H200 and at lower
per-workload rates on RTX~4090.}
\Description{Four grouped bar panels in two rows: the top row is
Qwen3-32B on four H200 GPUs at low and at load, the bottom row
repeats the same four workloads with Qwen3-8B on four commodity
RTX4090 GPUs. Each panel shows latency normalized to the no-reuse
floor for DynBranch, its C1-only arm, DSP, SPAgent, InstCache,
Helium, and Parrot across Routing, ReAct, Sub-Agent, and HCI.
A DynBranch arm is the lowest bar in every group of every panel.}
\label{fig:e4-mutpanels}
\end{figure*}
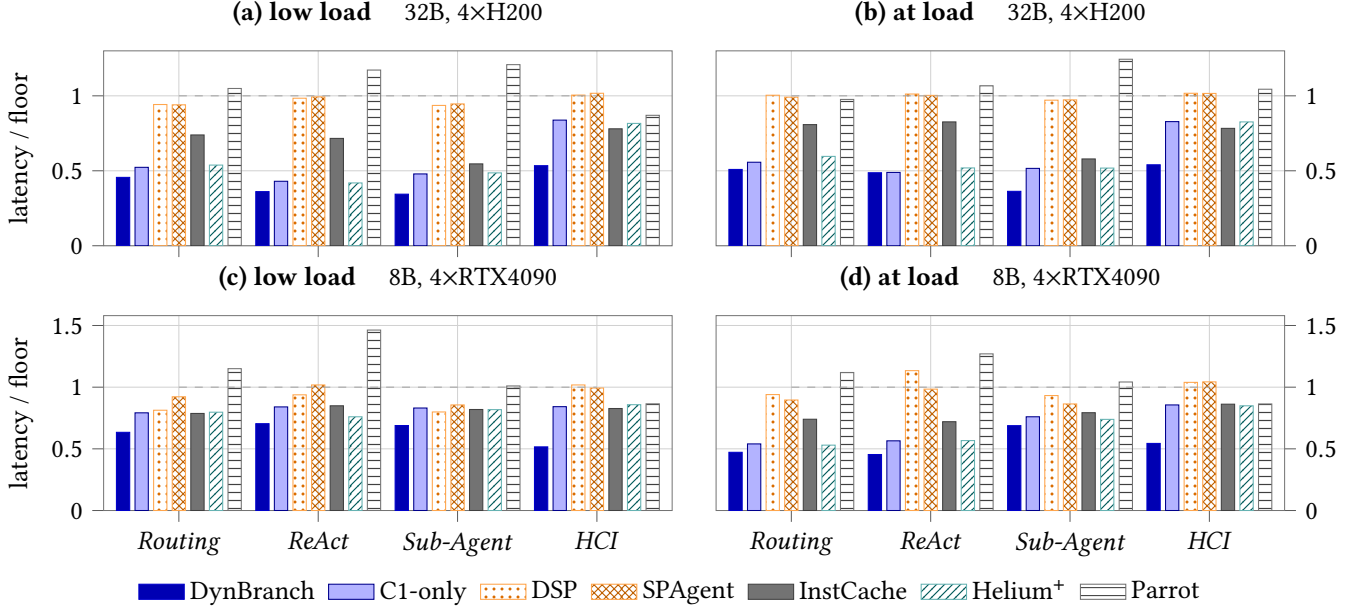

%% file: sections/eval.tex
\section{Evaluation}
\label{sec:eval}

We evaluate end-to-end latency, completion, and tail
(\S\ref{sec:eval:e2e}); how width, replicas, and depth trade GPU work for
latency (\S\ref{sec:eval:cost}); what reuse (C1), early fill (C2), and
each admission gate contribute (\S\ref{sec:eval:c1}); and how the system
holds up under rising and bursty load (\S\ref{sec:eval:c1c2-ablation}).
Task accuracy is checked against the floor in every experiment.

\subsection{Experimental Setup}

\noindent\textbf{Testbed.}
Our main testbed is a dedicated server with two AMD EPYC 9755 CPUs,
2.2\,TiB of main memory, and four NVIDIA H200 NVL GPUs with 141\,GiB
of HBM each.
Workers serve Qwen3-32B~\cite{qwen3} unless stated otherwise, and a
distilled Qwen3-0.6B drafts branch candidates.
The commodity replication in Figure~\ref{fig:e4-mutpanels}(c,d)
uses four distributed worker nodes on vast.ai, each serving Qwen3-8B
on one NVIDIA RTX 4090 with $24$\,GiB of device memory.

\vspace{0.05in}\noindent\textbf{Workloads.}
We instantiate the regimes in \S\ref{sec:bg:trace-shape} with:
\begin{list}{\labelitemi}{%
  \setlength{\leftmargin}{1.2em}%
  \setlength{\labelwidth}{0.8em}%
  \setlength{\labelsep}{0.4em}%
  \setlength{\itemindent}{0pt}%
  \setlength{\listparindent}{0pt}%
  \setlength{\itemsep}{0pt}%
  \setlength{\parsep}{0pt}%
  \setlength{\topsep}{2pt}}%
  \sloppy
  \item \textbf{\wl{Routing}:} BFCL function-calling records~\cite{patil2025bfcl},
        replayed with their handler choices and arguments.
  \item \textbf{\wl{ReAct}:} multi-hop questions from
        \textsc{MuSiQue}~\cite{musique} over a local HotpotQA-derived
        retrieval corpus~\cite{hotpotqa}.
  \item \textbf{\wl{Sub-Agent}:} \textsc{MuSiQue}-derived research sub-topics;
        the reuse-only run (Table~\ref{tab:c1-results}) uses the same entity distribution.
  \item \textbf{\wl{HCI}:} menu-offer/user-pick pairs from Schema-Guided
        Dialogue~\cite{sgd}, with paired, sampled think gaps between offer
        and selection~\cite{doherty,klm}; latency is measured on the user
        turn that follows the gap.
\end{list}

\vspace{0.05in}\noindent\textbf{Baselines.}
We compare five systems spanning completed-result
reuse~\cite{instcache,helium}, workflow scheduling, and
agent-level speculation. InstCache~\cite{instcache} pre-populates an
exact request cache from predicted instructions.
Helium$^{+}$~\cite{helium}\footnote{Helium$^{+}$ runs Helium's cache-aware
scheduling policy on the concurrent dispatch path and adds exact
per-operator result caching.} combines cache-aware scheduling
with exact operator-result reuse. Parrot~\cite{parrot} exposes workflow
dataflow as semantic variables. DSP~\cite{dsap} tunes speculative-plan
depth under a latency--cost objective. SPAgent~\cite{spagent} speculates
the next search-agent action and retains it in an in-request
buffer.

\vspace{0.05in}\noindent\textbf{Metrics.}
The reference throughout is a \emph{no-reuse floor}: the same stack with
speculation and cross-request reuse disabled. Open-loop experiments report
client-observed mean per-request wall time at arrival rate $\lambda$
(req/s) as the primary metric, with p99 characterizing the tail; serial comparisons report mean change against that
floor. Correctness is task accuracy relative to it; the absolute level is
workload-specific and not comparable across workloads. Error bars are
$\pm1$ s.d.\ over seeds.

\pgfplotsset{
  evalaxis/.style={
    axis lines=box, axis line style={black!55, line width=0.25pt},
    xtick pos=bottom,
    tick style={black!65, line width=0.35pt}, tick align=outside,
    tick label style={font=\normalsize}, label style={font=\normalsize},
    title style={font=\normalsize\bfseries, yshift=-1pt},
    xmajorgrids, ymajorgrids, grid style={black!18, line width=0.3pt},
  },
  evallegend/.style={
    font=\normalsize, draw=none, fill=none, inner sep=1pt,
    /tikz/every even column/.append style={column sep=4pt},
  },
  evalfloor/.style={color=black!78, mark=square*, mark size=1.35pt, thick},
  evalc1/.style={color=blue!45!black, mark=triangle*, mark size=1.55pt, thick},
  evalours/.style={color=blue!75!black, mark=*, mark size=1.45pt, very thick},
  evalsecondary/.style={color=orange!85!black, mark=square*, mark size=1.25pt,
    thick, dashed},
  evalreference/.style={color=black!48, mark=none, thick, densely dashed},
}

\subsection{End-to-End Comparison}
\label{sec:eval:e2e}

\sys{} achieves the lowest mean latency across all four workloads,
both load levels, and both hardware configurations
(Figure~\ref{fig:e4-mutpanels}). On 32B/H200, it improves over each
workload's strongest external baseline by $14$--$32\%$ at low load.
At load, the reductions are $15$--$31\%$ on \wl{Routing}, \wl{Sub-Agent}, and
\wl{HCI}, and $6.2\%$ on \wl{ReAct}. Across these cells, reductions against
the no-reuse floor span $46$--$66\%$. Task accuracy tracks the floor in
every experiment below: promotion is lossless by construction
(\S\ref{sec:design:lifecycle}).

\vspace{0.05in}\noindent\textbf{Early fill under load.}
C2 retains substantial benefit under load on \wl{Sub-Agent} and \wl{HCI},
reducing latency by $30\%$ and $35\%$ over C1-only, respectively;
\wl{Routing} gains $8.6\%$. On open-vocabulary \wl{ReAct}, the gain falls from
$16.0\%$ at low load to $0.4\%$ at load, while $53$--$55\%$ of
issued fills fail to match the resolved branch and are discarded at
both points. The similar mismatch rates accompany different latency
gains: under load the binding constraint is how much execution a fill can
still hide, not how often the branch is guessed right
(\S\ref{sec:eval:c1c2-ablation}).

\vspace{0.05in}\noindent\textbf{Completion and tail latency.}
Every arm completes the offered load with zero runtime failures.
Relative to the no-reuse floor on 32B/H200, \sys{} lowers p99 by
$11$--$33\%$ at load and $19$--$36\%$ at low load on \wl{Routing},
\wl{ReAct}, and \wl{Sub-Agent}; \wl{HCI} matches the floor at both points.

\vspace{0.05in}\noindent\textbf{Commodity replication.}
The benefit generalizes to the smaller-model, distributed deployment:
with Qwen3-8B on four RTX 4090s, \sys{} lowers mean latency by
$6.9$--$37.6\%$ relative to each cell's strongest external baseline
and $29.6$--$54.6\%$ relative to its no-reuse floor
(Figure~\ref{fig:e4-mutpanels}c,d). C2 improves over C1-only in every cell.

\subsection{Speculation Cost and Sensitivity}
\label{sec:eval:cost}

Width, replicas, and depth each buy latency with GPU work, and each
saturates; the benefit holds across model sizes and backbone families.

\vspace{0.05in}\noindent\textbf{Width.}
Wider fanout keeps lowering latency, but a growing share of what it adds is
thrown away.
At four workers, increasing fanout $K$ from $0$ to $3$ lowers mean
latency from $18.2$ to $15.2$\,s and adds $19.8$ GPU-seconds of
speculative work per request (Figure~\ref{fig:scale}a).
The reduction becomes statistically significant at $K{=}2$
(Welch's $t$-test~\cite{welch1947} against $K{=}0$).
The discarded share rises from $20\%$ at $K{=}1$ to $33\%$ at $K{=}3$.
Waste shifts from selected fills that fail to promote to fills for
unselected candidates.

Cancellation releases capacity promptly but cannot recover spent work:
abort is confirmed $3$\,ms after
branch resolution, by which point a cancelled fill has consumed $4.1$\,s of
engine time against $6.2$--$6.5$\,s for a completed one.

\vspace{0.05in}\noindent\textbf{Replicas.}
Replicas stop paying after the third worker.
Scaling from one to four workers with $K{=}W{-}1$ at fixed per-GPU load
reduces mean latency from $22.3$ to $15.2$\,s
(Figure~\ref{fig:scale}b). The fourth
worker saves only $0.25$\,s while adding
$1.9$ GPU-seconds of discarded work per request. Consumed work is
nearly flat once width reaches the workload's average of $2.67$
candidates per branch.
Task accuracy shows no regression across the seven width and replica
configurations.

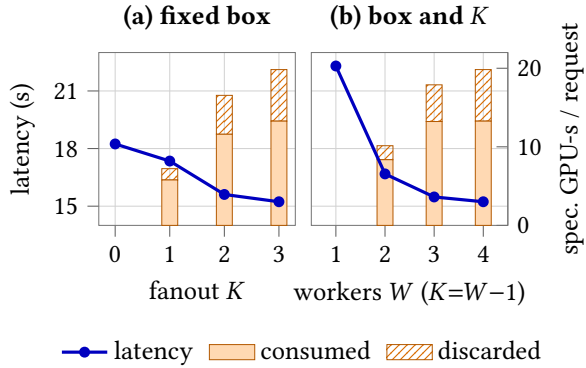
\begin{figure}[b]
\centering
\pgfplotsset{
  usedbar/.style={fill=orange!30, draw=orange!70!black, line width=0.3pt},
  wastebar/.style={pattern=north east lines, pattern color=orange!85!black,
    draw=orange!70!black, line width=0.3pt},
  set layers,
  spendaxis/.style={
    scale only axis, width=0.307\columnwidth, height=0.27\columnwidth,
    ybar stacked, bar width=6pt,
    ymin=0, ymax=22, ytick={0,10,20},
    axis y line*=right, axis x line=none, ytick pos=right, ylabel near ticks,
    tick label style={font=\normalsize}, label style={font=\normalsize},
    axis line style={black!55, line width=0.25pt},
    tick style={black!65, line width=0.35pt},
    y label style={text=black},
    yticklabel style={text=black}},
}
\begin{tikzpicture}
\begin{axis}[spendaxis, at={(0.150\columnwidth,0)}, anchor=south west,
  xmin=-0.3, xmax=3.3, xtick=\empty,
  ytick=\empty, axis y line=none]
\addplot[usedbar]  coordinates {(0,0)(1,5.78)(2,11.60)(3,13.27)};
\addplot[wastebar] coordinates {(0,0)(1,1.44)(2,4.95)(3,6.56)};
\end{axis}
\begin{axis}[evalaxis, at={(0.150\columnwidth,0)}, anchor=south west,
  scale only axis, width=0.307\columnwidth, height=0.27\columnwidth,
  title={(a) fixed box},
  xlabel={fanout $K$}, ylabel={latency (s)},
  ytick pos=left, ymin=14, ymax=23, ytick={15,18,21},
  xmin=-0.3, xmax=3.3, xtick={0,1,2,3},
  axis line style={black!55, line width=0.25pt},
  tick style={black!65, line width=0.35pt},
  y label style={text=black},
  yticklabel style={text=black},
  legend style={evallegend, at={(1.041,-0.62)}, anchor=north,
    legend columns=3}]
\addplot[evalours] coordinates {(0,18.24)(1,17.35)(2,15.61)(3,15.23)};
\addlegendentry{latency}
\addlegendimage{area legend, usedbar}\addlegendentry{consumed}
\addlegendimage{area legend, wastebar}\addlegendentry{discarded}
\end{axis}
\begin{axis}[spendaxis, at={(0.482\columnwidth,0)}, anchor=south west,
  xmin=0.5, xmax=4.5, xtick=\empty,
  ylabel={spec.\ GPU-s / request}]
\addplot[usedbar]  coordinates {(1,0)(2,8.36)(3,13.22)(4,13.27)};
\addplot[wastebar] coordinates {(1,0)(2,1.78)(3,4.67)(4,6.56)};
\end{axis}
\begin{axis}[evalaxis, at={(0.482\columnwidth,0)}, anchor=south west,
  scale only axis, width=0.307\columnwidth, height=0.27\columnwidth,
  title={(b) box and $K$},
  xlabel={workers $W$ ($K{=}W{-}1$)},
  ytick pos=left, ymin=14, ymax=23, ytick={15,18,21}, yticklabels={,,},
  ytick style={draw=none},
  xmin=0.5, xmax=4.5, xtick={1,2,3,4},
  axis line style={black!55, line width=0.25pt},
  tick style={black!65, line width=0.35pt},
  yticklabel style={text=black}]
\addplot[evalours] coordinates {(1,22.30)(2,16.68)(3,15.48)(4,15.23)};
\end{axis}
\end{tikzpicture}
\caption{Speculation width and scale on \wl{Routing}. (a)~Fixed four-worker
box, $K$ swept. (b)~Box and $K$ grown together ($K{=}W{-}1$).}
\Description{Two panels sharing a latency axis on the left and a
speculative GPU-seconds axis on the right. Panel a holds the box at four
workers and raises the speculative fanout K from 0 to 3: mean latency
falls from 18.2 to 15.2 seconds while speculative engine time rises from
zero to 19.8 GPU-seconds per request, of which the discarded part grows
from a fifth to a third. Panel b scales the box and K together from one
to four workers: latency falls from 22.3 to 15.2 seconds for zero, 10.1,
17.9 and 19.8 GPU-seconds per request. Both panels finish at the same point.}
\label{fig:scale}
\end{figure}

\vspace{0.05in}\noindent\textbf{Depth.}
On \wl{ReAct}, speculating beyond one run-ahead cycle (one hop ahead)
does not pay. One cycle captures the useful overlap in the hop chain,
reducing median session wall time by $11.8\%$ and serving $60$ of
$96$ arrivals. Two or three cycles reach the additional hops but
consume no results from them; latency changes by at most $0.7\%$,
within the observed $2$--$5\%$ run-to-run spread. Deeper speculative
decisions finish too late to precede canonical work.

\vspace{0.05in}\noindent\textbf{Repetition and model size.}
More recurrence increases relative gain; larger serving models
increase absolute savings (Figure~\ref{fig:e5-repeat}). Raising
exact replay from $0.22$ to $0.77$ increases mean-latency speedup
from $1.8\times$ to $8.6\times$. Across Qwen3 0.6B--32B models, speedup
stays at $1.75$--$2.06\times$ while savings grow from $1.6$ to
$8.9$\,s per request.

\begin{figure}[b]
\centering
\begin{tikzpicture}
\begin{axis}[evalaxis, at={(0.150\columnwidth,0)}, anchor=south west,
  scale only axis, width=0.240\columnwidth, height=0.24\columnwidth,
  title={(a) repeat rate}, xlabel={repeat}, ylabel={speedup vs.\ floor},
  xmajorgrids, ytick pos=left,
  ymin=1, ymax=9.2, xmin=0.15, xmax=0.82,
  legend style={evallegend, at={(1.237,-0.66)}, anchor=north,
    legend columns=2,
    /tikz/every even column/.append style={column sep=4pt}}]
\addplot[evalours,
         error bars/.cd, y dir=both, y explicit,
         error bar style={blue!75!black, line width=0.45pt},
         error mark options={blue!75!black, mark size=1.6pt, line width=0.45pt}] coordinates
  {(0.22,1.84) +- (0,0.12) (0.31,2.37) +- (0,0.10) (0.41,3.00) +- (0,0.04)
   (0.52,3.86) +- (0,0.21) (0.65,5.92) +- (0,0.72) (0.77,8.63) +- (0,1.49)};
\addlegendentry{speedup (left)}
\addlegendimage{evalsecondary}
\addlegendentry{seconds saved (right)}
\end{axis}
\begin{axis}[at={(0.150\columnwidth,0)}, anchor=south west,
  scale only axis, width=0.240\columnwidth, height=0.24\columnwidth,
  xmin=0.15, xmax=0.82, ymin=0, ymax=20,
  axis y line*=right, axis x line=none, ytick pos=right, ylabel near ticks,
  axis line style={black!55, line width=0.25pt},
  tick style={black!65, line width=0.35pt},
  y label style={font=\normalsize, text=black},
  yticklabel style={font=\normalsize, text=black}]
\addplot[evalsecondary, error bars/.cd, y dir=both, y explicit,
         error bar style={orange!70!black, line width=0.45pt},
         error mark options={orange!70!black, mark size=1.6pt, line width=0.45pt}] coordinates
  {(0.22,8.79) +- (0,1.01) (0.31,11.17) +- (0,0.36) (0.41,13.28) +- (0,0.21)
   (0.52,14.89) +- (0,1.06) (0.65,15.81) +- (0,0.44) (0.77,16.73) +- (0,1.16)};
\end{axis}
\begin{axis}[evalaxis, at={(0.600\columnwidth,0)}, anchor=south west,
  scale only axis, width=0.240\columnwidth, height=0.24\columnwidth,
  title={(b) model size}, xmode=log, log basis x=10,
  xlabel={params},
  ymin=1, ymax=2.4, xmin=0.5, xmax=40,
  xtick={0.6,4,32}, xticklabels={0.6B,4B,32B},
  xmajorgrids, ytick pos=left]
\addplot[black!42, thick, densely dotted, mark=none, forget plot] coordinates
  {(0.5,1.86)(40,1.86)};
\addplot[evalours,
         error bars/.cd, y dir=both, y explicit,
         error bar style={blue!75!black, line width=0.45pt},
         error mark options={blue!75!black, mark size=1.6pt, line width=0.45pt}] coordinates
  {(0.6,1.75) +- (0,0.09) (1.7,1.77) +- (0,0.10) (4,1.86) +- (0,0.14)
   (8,2.06) +- (0,0.20) (14,1.86) +- (0,0.15) (32,1.85) +- (0,0.14)};
\end{axis}
\begin{axis}[at={(0.600\columnwidth,0)}, anchor=south west,
  scale only axis, width=0.240\columnwidth, height=0.24\columnwidth,
  xmode=log, log basis x=10, xmin=0.5, xmax=40,
  ymin=0, ymax=11, axis y line*=right, axis x line=none,
  ytick pos=right, ylabel near ticks, ylabel={seconds saved (s)},
  axis line style={black!55, line width=0.25pt},
  tick style={black!65, line width=0.35pt},
  y label style={font=\normalsize, text=black},
  yticklabel style={font=\normalsize, text=black}]
\addplot[evalsecondary, error bars/.cd, y dir=both, y explicit,
         error bar style={orange!70!black, line width=0.45pt},
         error mark options={orange!70!black, mark size=1.6pt, line width=0.45pt}] coordinates
  {(0.6,1.61) +- (0,0.15) (1.7,2.26) +- (0,0.14) (4,3.01) +- (0,0.28)
   (8,3.54) +- (0,0.30) (14,3.79) +- (0,0.46) (32,8.91) +- (0,1.09)};
\end{axis}
\end{tikzpicture}
\caption{Sensitivity on \wl{Routing} to (a)~exact-replay rate and
(b)~serving-model size (dotted: mean over sizes).}
\Description{Two small panels. Panel a shows speedup over the
no-reuse floor rising monotonically from 1.8x to 8.6x as the
repeat rate grows. Panel b shows speedup flat near 1.9x across
model sizes while the absolute saved seconds grow from about 1.6
seconds at 0.6B to 8.9 seconds at 32B.}
\label{fig:e5-repeat}
\end{figure}
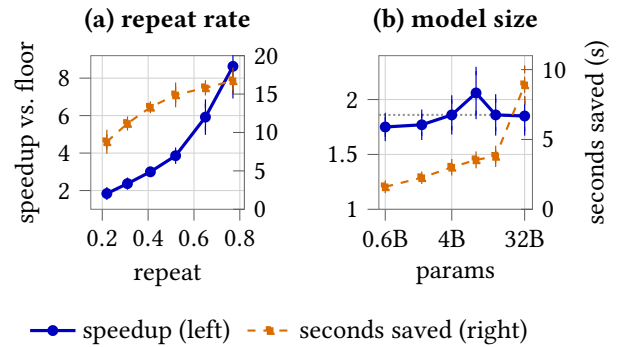

\vspace{0.05in}\noindent\textbf{Backbone family.}
Speedup varies little with the served family.
Four alternative backbones (Figure~\ref{fig:e5-family}a) yield
$1.55$--$1.99\times$ speedup over their paired no-reuse floors across
three seeds each, with task accuracy within
$0.012$ of the floor.

\vspace{0.05in}\noindent\textbf{Drafter variants.}
Higher branch-prediction accuracy does not by itself improve end-to-end
latency: a prediction must also yield its result in time for demand.
Across 0.6B--8B drafters, mean-latency speedup
stays at $1.63$--$1.70\times$, with no statistically significant
difference between sizes across three seeds; every size improves over the
floor. The 8B drafter has the highest
top-1 accuracy over $140$ held-out branches
(Figure~\ref{fig:e5-family}b) and issues the most fills,
but has the lowest consumption rate, $50\%$ versus $70\%$ for 1.7B.

The drafter need not match the served model's family:
Llama-3.2-1B matches Qwen3-1.7B in top-1 accuracy, and Qwen2.5-3B
exceeds Qwen3-4B. Distilling 32B branch decisions into a 0.6B LoRA
drafter yields $0.34$ top-1 accuracy, exceeding the 8B drafter at a
little over half its decode cost (\texttt{+LoRA} in
Figure~\ref{fig:e5-family}b).

\begin{figure}[t]
\centering
\begin{tikzpicture}
\begin{axis}[evalaxis, at={(0.150\columnwidth,0)}, anchor=south west,
  scale only axis, width=0.3225\columnwidth, height=0.232\columnwidth,
  title={(a) backbone},
  ylabel={speedup vs.\ floor},
  symbolic x coords={Qwen3,Gemma3,OLMo2,Llama3,Yi1.5},
  xtick=data, xticklabel style={font=\normalsize, rotate=90, anchor=east},
  enlarge x limits=0.17, ymin=1, ymax=2.4, ytick={1,1.4,1.8,2.2},
  xmajorgrids=false, ytick pos=left]
\addplot[ybar, bar width=7pt, fill=fmwin, draw=orange!85!black, line width=0.3pt,
         error bars/.cd, y dir=both, y explicit,
         error bar style={black!70, line width=0.45pt},
         error mark options={black!70, mark size=1.6pt, line width=0.45pt}] coordinates
  {(Qwen3,1.84) +- (0,0.13) (Gemma3,1.93) +- (0,0.23) (OLMo2,1.99) +- (0,0.16)
   (Llama3,1.89) +- (0,0.18) (Yi1.5,1.55) +- (0,0.09)};
\addplot[evalreference, forget plot] coordinates {(Qwen3,1.84) (Yi1.5,1.84)};
\end{axis}
\begin{axis}[evalaxis, at={(0.498\columnwidth,0)}, anchor=south west,
  scale only axis, width=0.3225\columnwidth, height=0.232\columnwidth,
  title={(b) drafter}, xmode=log, log basis x=10,
  xlabel={drafter params}, ylabel={branch accuracy},
  xtick={0.6,4,8}, xticklabels={0.6B,4B,8B},
  xticklabel style={font=\normalsize},
  xmin=0.35, xmax=13, ymin=0.05, ymax=0.55,
  xmajorgrids, ytick={0.1,0.3,0.5},
  ytick pos=right, yticklabel pos=right, ylabel near ticks,
  ylabel style={font=\normalsize}, yticklabel style={font=\normalsize},
  axis line style={black!55, line width=0.25pt},
  tick style={black!65, line width=0.35pt}]
\addplot[evalc1] coordinates {(0.6,0.300)(1.7,0.207)(4,0.271)(8,0.414)};
\addplot[evalours] coordinates {(0.6,0.186)(1.7,0.171)(4,0.207)(8,0.314)};
\addplot[color=orange!85!black, mark=triangle*, mark size=1.9pt, only marks]
  coordinates {(0.6,0.457)};
\addplot[color=orange!85!black, mark=*, mark size=1.7pt, only marks]
  coordinates {(0.6,0.336)};
\node[font=\normalsize, text=orange!85!black, anchor=west] at (axis cs:0.80,0.505) {+LoRA};
\node[font=\normalsize, text=blue!45!black, anchor=west] at (axis cs:1.55,0.40) {top-3};
\node[font=\normalsize, text=blue!75!black, anchor=west] at (axis cs:1.55,0.105) {top-1};
\end{axis}
\end{tikzpicture}
\caption{Sensitivity to the models. (a)~Served backbone on \wl{Routing}
(dashed: Qwen3-32B). (b)~Drafter size on \wl{ReAct}.}
\Description{Two small panels. Panel a plots mean-latency speedup for
five served backbones, all between 1.55 and 1.99 times the floor, with
Yi-1.5-34B lowest. Panel b plots the drafter's top-1 and top-3 branch
accuracy against parameter count on a log axis, both rising to their
highest values at 8B; two isolated points mark the 0.6B drafter with a
LoRA adapter, sitting above the 8B point on both cut-offs.}
\label{fig:e5-family}
\end{figure}
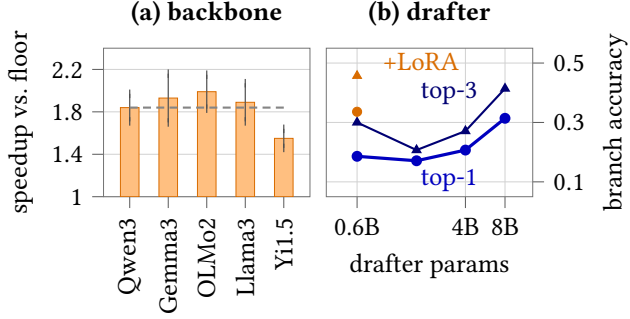

\subsection{Ablation Study}
\label{sec:eval:c1}

C1 captures reuse across partially overlapping requests; C2 lowers
latency even with cross-request reuse disabled. Both preserve task
accuracy against their paired no-reuse floors. A final pair of
experiments ablates the two admission gates.

\vspace{0.05in}\noindent\textbf{Reuse alone.}
With C2 disabled, \sys{} reduces mean latency on partial-overlap
requests by $31\%$, matching Helium$^{+}$'s $30\%$ reduction;
InstCache~\cite{instcache} gains only $7\%$
(Table~\ref{tab:c1-results}). The workload mixes exact repeats,
shared sub-DAGs, and fresh negative controls; arms are paired by request
id on cold-started stacks. Exact repeats reach the cache-hit path; fresh requests yield
no reuse benefit. Across the mix, \sys{} reduces mean latency by
$32.2\%$, with zero unsafe splices. C1 thus captures reuse below the whole-request
boundary while retaining write-cascade invalidation.

\begin{table}[t]
\centering
\begin{tabular}{lrrr}
\toprule
\textbf{Arm} & \textbf{repeat} & \textbf{partial overlap} & \textbf{fresh} \\
\midrule
floor & $14.4$\,s & $14.8$\,s & $13.6$\,s \\
InstCache & $0.01$\,s & $13.8$\,s\ ($-7\%$) & $14.6$\,s \\
Helium$^{+}$ & $0.16$\,s & $10.3$\,s\ ($-30\%$) & $14.0$\,s \\
\textbf{\sys{}} & $\mathbf{0.01}$\,\textbf{s} & $\mathbf{10.2}$\,\textbf{s}\ ($\mathbf{-31\%}$) & $14.1$\,s \\
\bottomrule
\end{tabular}
\caption{Reuse alone by request type (Qwen3-32B, paired by request id);
\emph{partial overlap} = a sub-DAG shared with another record.}
\label{tab:c1-results}
\end{table}

\vspace{0.05in}\noindent\textbf{Early fill alone.}
With C1 disabled, \sys{} reduces mean latency by $8.5$--$47.7\%$
against the shared no-reuse floor across all four workloads on both
two- and four-GPU deployments (Table~\ref{tab:c2-results}). The baselines run their native admission, and
\sys{} outperforms DSP~\cite{dsap} and SPAgent~\cite{spagent}
throughout. These gains overlap with C1's: some filled stages would
also be served by reuse.

The gain extends beyond a single predicted action. On \wl{Routing},
$349$ fills yield $500$ consumed stages through chained
promotion. \wl{HCI} shows the largest separation: completed fills serve
$73\%$ of choices, cutting user-turn wait by $46.6$--$47.7\%$;
neither baseline materially benefits from the think gap.

\begin{table}[t]
\centering
\begin{tabular}{lrrrr}
\toprule
\textbf{System} & \textbf{\wl{Routing}} & \textbf{\wl{ReAct}} & \textbf{\wl{Sub-Agent}} & \textbf{\wl{HCI}} \\
\midrule
\multicolumn{5}{@{}l}{\emph{2 GPUs}} \\
\textbf{\sys} & $\mathbf{-18.6}$ & $\mathbf{-8.7}$ & $\mathbf{-34.4}$ & $\mathbf{-46.6}$ \\
SPAgent & $-12.3$ & $-0.7$ & $-23.6$ & $-1.0$ \\
DSP & $-13.1$ & $-0.4$ & $-23.8$ & $-0.7$ \\
\midrule
\multicolumn{5}{@{}l}{\emph{4 GPUs}} \\
\textbf{\sys} & $\mathbf{-21.9}$ & $\mathbf{-8.5}$ & $\mathbf{-34.0}$ & $\mathbf{-47.7}$ \\
SPAgent & $-10.6$ & $+0.7$ & $-25.0$ & $+1.6$ \\
DSP & $-11.2$ & $-0.2$ & $-24.7$ & $-0.2$ \\
\bottomrule
\end{tabular}
\caption{C2 isolation (serial closed-loop): mean $\Delta\%$ against the
shared no-reuse floor, C1 off on every arm.}
\label{tab:c2-results}
\end{table}

\vspace{0.05in}\noindent\textbf{Case study.}
Early fill removes four downstream stages from the critical path in
the paired \wl{Routing} request in Figure~\ref{fig:case}, saving $5.75$\,s
end to end. The request computes an area and then converts its units.
Two fills overlap router decoding on other workers, each completing
a handler and its executor before
demand, so they are served at zero engine time. The floor instead runs
all four after the decisions that gate them, for $6.50$\,s in total,
close to the observed end-to-end saving.

Both arms take the same route, and this request has the median paired
saving. Across the run, $99\%$ of promotions find a completed fill
and $1\%$ wait for one. Producer identities distinguish speculative
consumption from ordinary cache hits; attributed stage counts match
the system's own counter.

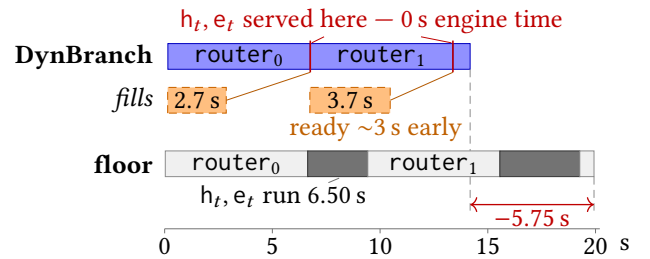
\begin{figure}[b]
\centering
\begin{tikzpicture}[x=0.285cm, y=1cm,
  dec/.style={fill=blue!43, draw=blue!75!black, line width=0.3pt},
  fdec/.style={fill=fmpale, draw=black!45, line width=0.3pt},
  fexec/.style={fill=fmexec, draw=black!80, line width=0.3pt},
  specfill/.style={fill=fmwin, draw=orange!80!black, line width=0.4pt, densely dashed},
  avoided/.style={draw=red!75!black, line width=0.6pt},
  lbl/.style={font=\normalsize, inner sep=0.5pt},
  lane/.style={font=\normalsize, anchor=east, inner sep=1pt}]
  \def\yD{1.30}\def\yF{0.72}\def\yL{-0.12}\def\hh{0.17}

  \node[lane] at (-0.4,\yD) {\textbf{\sys{}}};
  \fill[dec] (0.15,\yD-\hh) rectangle (6.76,\yD+\hh);
  \node[lbl] at (3.46,\yD) {\texttt{router$_0$}};
  \fill[dec] (6.76,\yD-\hh) rectangle (13.38,\yD+\hh);
  \node[lbl] at (10.07,\yD) {\texttt{router$_1$}};
  \fill[dec] (13.38,\yD-\hh) rectangle (14.18,\yD+\hh);
  \foreach \t in {6.76,13.38} {\draw[avoided] (\t,\yD-\hh-0.04) -- (\t,\yD+\hh+0.04);}
  \node[lbl, red!75!black, anchor=south] at (9.6,\yD+0.28)
    {\texttt{h}$_t$,\,\texttt{e}$_t$ served here --- $0$\,s engine time};
  \draw[red!75!black, line width=0.3pt] (7.4,\yD+0.26) -- (6.83,\yD+\hh+0.01);
  \draw[red!75!black, line width=0.3pt] (11.9,\yD+0.26) -- (13.31,\yD+\hh+0.01);

  \node[lane] at (-0.4,\yF) {\textit{fills}};
  \fill[specfill] (0.14,\yF-\hh) rectangle (2.86,\yF+\hh);
  \node[lbl] at (1.50,\yF) {$2.7$\,s};
  \fill[specfill] (6.74,\yF-\hh) rectangle (10.47,\yF+\hh);
  \node[lbl] at (8.61,\yF) {$3.7$\,s};
  \foreach \f/\t in {2.86/6.76, 10.47/13.38} {
    \draw[orange!70!black, line width=0.3pt]
      (\f,\yF) -- (\t,\yD-\hh-0.04);}
  \node[lbl, orange!75!black, anchor=north east] at (13.9,\yF-\hh-0.02)
    {ready $\sim$3\,s early};

  \node[lane] at (-0.4,\yL) {\textbf{floor}};
  \fill[fdec] (0.03,\yL-\hh) rectangle (6.64,\yL+\hh);
  \node[lbl] at (3.34,\yL) {\texttt{router$_0$}};
  \fill[fexec] (6.64,\yL-\hh) rectangle (9.43,\yL+\hh);
  \fill[fdec] (9.43,\yL-\hh) rectangle (15.56,\yL+\hh);
  \node[lbl] at (12.50,\yL) {\texttt{router$_1$}};
  \fill[fexec] (15.56,\yL-\hh) rectangle (19.27,\yL+\hh);
  \fill[fdec] (19.27,\yL-\hh) rectangle (19.93,\yL+\hh);
  \node[lbl, anchor=north] at (5.5,\yL-0.26) {\texttt{h}$_t$,\,\texttt{e}$_t$ run $6.50$\,s};
  \draw[black!55, line width=0.3pt] (7.6,\yL-0.25) -- (8.1,\yL-\hh-0.02);

  \draw[black!45, densely dashed, line width=0.3pt] (14.18,\yD-\hh) -- (14.18,\yL+\hh);
  \draw[black!45, densely dashed, line width=0.3pt] (14.18,\yL-\hh) -- (14.18,\yL-0.62);
  \draw[black!45, densely dashed, line width=0.3pt] (19.93,\yL-\hh) -- (19.93,\yL-0.62);
  \draw[<->, red!75!black, line width=0.5pt] (14.18,\yL-0.54) -- (19.93,\yL-0.54);
  \node[lbl, red!75!black, anchor=north] at (17.05,\yL-0.58) {$-5.75$\,s};

  \draw[black!55, line width=0.3pt] (0,\yL-0.86) -- (20.2,\yL-0.86);
  \foreach \t in {0,5,10,15,20} {
    \draw[black!55, line width=0.3pt] (\t,\yL-0.86) -- (\t,\yL-0.92);
    \node[lbl, anchor=north, inner sep=1.5pt] at (\t,\yL-0.92) {\t};}
  \node[lbl, anchor=north west, inner sep=1.5pt] at (20.9,\yL-0.92) {s};
\end{tikzpicture}
\caption{A paired \wl{Routing} request. Each decision (\texttt{router$_t$})
gates a handler and executor (\texttt{h}$_t$,\,\texttt{e}$_t$).}
\Description{Three lanes over 20 seconds. The system lane runs two router
stages back to back and an aggregator, finishing at 14.2 seconds, with
ticks where the four handler and executor stages were served without running. A
middle lane shows two speculative fills of 2.7 and 3.7 seconds, each beginning
as the router before it starts decoding and finishing about three seconds
before that router ends, on different workers. The floor lane runs the same two
routers plus all four handler and executor stages and finishes at 19.9 seconds.
An arrow between the two completion times is labeled minus 5.75 seconds.}
\label{fig:case}
\end{figure}

\vspace{0.05in}\noindent\textbf{Two-level attribution.}
Pricing both levels issues $8.3\times$ fewer fills than disabling
Action~0 and fixing Action~1, at indistinguishable latency
(Figure~\ref{fig:twolevel}a at $\lambda{=}0.80$; $n{=}155$, paired
within batch with arm order reversed between batches). Each level controls a distinct cost:
Action~0 cuts candidate-generation openings from $1.53$ to $0.35$ per request with
Action~1 priced, and from $1.42$ to $0.50$ with it fixed; Action~1
reduces fill execution under either setting of Action~0.

\vspace{0.05in}\noindent\textbf{Generation cost.}
Only Action~0 removes drafting cost; a fill threshold cannot. On
\wl{ReAct} with C1 disabled, we disable Action~0, so the drafter runs at
every branch, and replace Action~1 with fixed thresholds $p^*$. Across
these thresholds, fills per request vary $55\times$ ($0.09$--$4.99$), yet
drafting stays at $2.5$--$2.6$\,s per request (Figure~\ref{fig:twolevel}b). Against
$p^*{=}0.999$, which issues almost no fills, the fixed thresholds lower
latency by $2.5$--$3.5\%$, within noise (their CIs include zero); \sys{}
with both gates priced lowers it by $4.5\%$ (95\% CI $0.5$--$8.5\%$,
$n{=}174$ over three order-balanced seeds at $\lambda{=}0.12$) while
drafting $0.89$\,s, $2.8\times$ less. Task accuracy spans
$0.305$--$0.328$ across the sweep.

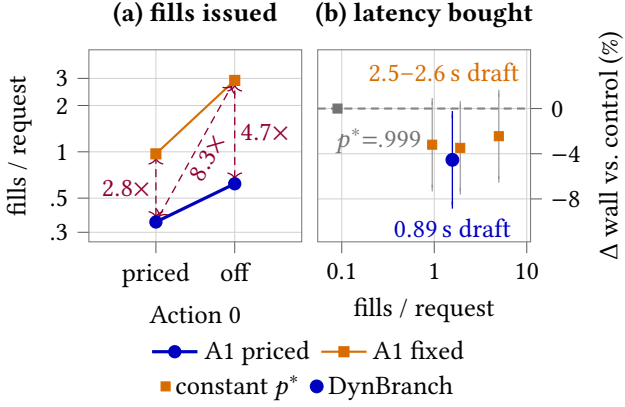
\begin{figure}[t]
\centering
\begin{tikzpicture}[ratio/.style={draw=purple!75!black, densely dashed}]
\begin{axis}[evalaxis, at={(0.150\columnwidth,0)}, anchor=south west,
  scale only axis, width=0.3325\columnwidth, height=0.30\columnwidth,
  title={(a) fills issued},
  xlabel={Action 0}, ylabel={fills / request}, ymode=log,
  ytick pos=left,
  symbolic x coords={A0 priced,A0 off}, xtick=data,
  xticklabels={priced,off},
  xticklabel style={font=\normalsize},
  ymin=0.26, ymax=4.6, ytick={0.3,0.5,1,2,3},
  yticklabels={.3,.5,1,2,3}, log ticks with fixed point,
  enlarge x limits=0.85,
  legend style={evallegend, at={(1.038,-0.47)}, anchor=north, legend columns=3,
    /tikz/every even column/.append style={column sep=2pt}}]
\draw[<->, ratio, line width=0.5pt]
  (axis cs:A0 priced,0.37) -- (axis cs:A0 priced,0.92);
\node[font=\normalsize, text=purple!75!black, anchor=east, inner sep=2pt]
  at (axis cs:A0 priced,0.583) {$2.8\times$};
\draw[<->, ratio, line width=0.5pt]
  (axis cs:A0 off,0.66) -- (axis cs:A0 off,2.74);
\node[font=\normalsize, text=purple!75!black, anchor=west, inner sep=2pt]
  at (axis cs:A0 off,1.345) {$4.7\times$};
\addplot[color=blue!75!black, mark=*, mark size=1.9pt, very thick]
  coordinates {(A0 priced,0.35) (A0 off,0.62)};
\addlegendentry{A1 priced}
\addplot[color=orange!85!black, mark=square*, mark size=1.7pt, thick]
  coordinates {(A0 priced,0.97) (A0 off,2.91)};
\addlegendentry{A1 fixed}
\draw[<->, ratio, line width=0.5pt, shorten >=2pt, shorten <=2pt]
  (axis cs:A0 priced,0.35) -- (axis cs:A0 off,2.91)
  node[midway, sloped, below, font=\normalsize, text=purple!75!black,
       inner sep=1.6pt] {$8.3\times$};
\end{axis}
\begin{axis}[evalaxis, at={(0.5075\columnwidth,0)}, anchor=south west,
  scale only axis, width=0.3325\columnwidth, height=0.30\columnwidth,
  title={(b) latency bought},
  xlabel={fills / request}, ylabel={$\Delta$ wall vs.\ control (\%)},
  xmode=log, log basis x=10,
  ytick pos=right, yticklabel pos=right, ylabel near ticks,
  ylabel style={font=\normalsize}, yticklabel style={font=\normalsize},
  axis line style={black!55, line width=0.25pt},
  tick style={black!65, line width=0.35pt},
  xmin=0.055, xmax=11, ymin=-11.8, ymax=5.2,
  xtick={0.1,1,10}, xticklabels={0.1,1,10},
  xticklabel style={font=\normalsize},
  ytick={-8,-4,0},
  legend style={evallegend, at={(-0.0376,-0.65)}, anchor=north, legend columns=2,
    /tikz/every even column/.append style={column sep=2pt}}]
\addplot[evalreference, forget plot] coordinates {(0.055,0)(11,0)};
\addplot[color=orange!85!black, mark=square*, mark size=1.7pt, only marks,
         error bars/.cd, y dir=both, y explicit,
         error bar style={black!45, line width=0.4pt},
         error mark options={black!45, mark size=1.3pt, line width=0.4pt}]
  coordinates {(0.95,-3.20) +- (0,3.83) (1.91,-3.48) +- (0,3.83)
               (4.99,-2.45) +- (0,3.80)};
\addlegendentry{constant $p^*$}
\addplot[color=blue!75!black, mark=*, mark size=2.4pt, only marks,
         error bars/.cd, y dir=both, y explicit,
         error bar style={blue!75!black, line width=0.5pt},
         error mark options={blue!75!black, mark size=1.4pt, line width=0.5pt}]
  coordinates {(1.57,-4.54) +- (0,3.99)};
\addlegendentry{\sys}
\addplot[color=black!55, mark=square*, mark size=1.7pt, only marks, forget plot]
  coordinates {(0.09,0)};
\node[font=\normalsize, text=black!55, anchor=north west] at (axis cs:0.070,-0.9)
  {$p^*{=}.999$};
\node[font=\normalsize, text=orange!85!black, anchor=north east] at (axis cs:10.5,4.9)
  {$2.5$--$2.6$\,s draft};
\node[font=\normalsize, text=blue!75!black, anchor=north] at (axis cs:1.57,-8.9)
  {$0.89$\,s draft};
\end{axis}
\end{tikzpicture}
\caption{Two-level admission on \wl{ReAct}. (a)~The $2\times2$ ablation.
(b)~The same gate vs.\ fixed thresholds (95\% CIs).}
\Description{Two panels. The left plots fills per request on a
logarithmic axis for the two by two of Action 0 and Action 1, as two
lines: with Action 1 priced, 0.35 rising to 0.62 as Action 0 is
disabled; with Action 1 fixed, 0.97 rising to 2.91. The vertical gaps
are 2.8 and 4.7 times. The right panel plots paired wall time against
the p-star 0.999 control on a logarithmic fill axis: the remaining arms
sit below the control by 0.5 to 0.95 seconds, the system lowest at
0.95 seconds and the only one whose confidence interval clears zero,
with the intervals of the firing arms overlapping each other.}
\label{fig:twolevel}
\end{figure}

\subsection{Stress Test}
\label{sec:eval:c1c2-ablation}

Under rising and bursty load, priced admission keeps the benefit that
fixed policies lose.

\vspace{0.05in}\noindent\textbf{Load decomposition.}
C1's saving follows the request mix rather than the rate; C2's follows
load.
Across five open-loop \wl{Routing} rates, C1 removes $40$--$48\%$ of
the floor latency; C2 removes a further
$0.25$--$1.39$\,s (Figure~\ref{fig:c1c2-ablation}). The floor,
C1-only, and full arms replay the same schedule, giving
$\Delta_{C1}=\text{floor}-\text{C1only}$ and
$\Delta_{C2}=\text{C1only}-\text{full}$; the three stack to the no-reuse
floor in Figure~\ref{fig:c1c2-ablation}a. All requests complete. C1-only and full have identical task accuracy.

C2's gain falls by $82\%$ from its peak as load rises but remains
positive at every measured rate. Fills issued fall from $1.4$ to
$0.9$ per request, while latency saved per fill falls by $72\%$;
the latter accounts for about three quarters of the narrowing.
The gain therefore depends on both admission volume and how much
execution each fill hides within the remaining resolution window.

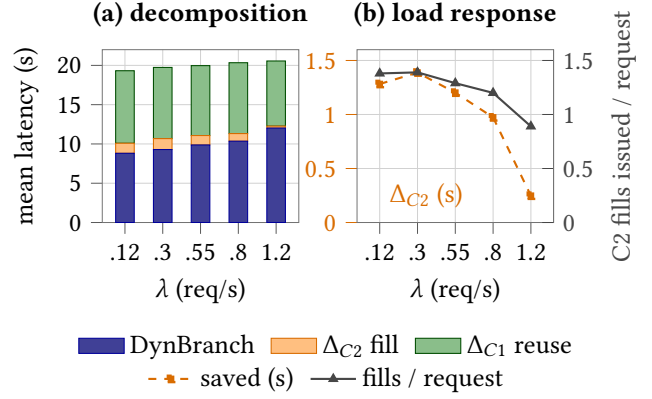
\begin{figure}[t]
\centering
\begin{tikzpicture}
\begin{axis}[evalaxis, at={(0.130\columnwidth,0)}, anchor=south west,
  scale only axis, width=0.307\columnwidth, height=0.27\columnwidth,
  ybar stacked, bar width=7pt,
  title={(a) decomposition},
  ytick pos=left,
  symbolic x coords={0.12,0.3,0.55,0.8,1.2}, xtick=data,
  xticklabels={.12,.3,.55,.8,1.2},
  xlabel={$\lambda$ (req/s)}, ylabel={mean latency (s)},
  ymin=0, ymax=22.00, ytick={0,5,10,15,20},
  enlarge x limits=0.15, area legend,
  legend style={evallegend, at={(1.148,-0.60)}, anchor=north,
    legend columns=3,
    /tikz/every even column/.append style={column sep=4pt}}]
\addplot[fill=blue!45!black!75!white, draw=blue!60!black] coordinates
  {(0.12,8.81)(0.3,9.29)(0.55,9.87)(0.8,10.35)(1.2,12.02)};
\addlegendentry{\sys{}}
\addplot[fill=fmwin, draw=orange!85!black] coordinates
  {(0.12,1.28)(0.3,1.39)(0.55,1.20)(0.8,0.97)(1.2,0.25)};
\addlegendentry{$\Delta_{C2}$ fill}
\addplot[fill=fmreuse, draw=green!30!black] coordinates
  {(0.12,9.23)(0.3,9.06)(0.55,8.90)(0.8,9.02)(1.2,8.29)};
\addlegendentry{$\Delta_{C1}$ reuse}
\end{axis}
\begin{axis}[evalaxis, at={(0.528\columnwidth,0)}, anchor=south west,
  scale only axis, width=0.307\columnwidth, height=0.27\columnwidth,
  title={(b) load response},
  symbolic x coords={0.12,0.3,0.55,0.8,1.2}, xtick=data,
  xticklabels={.12,.3,.55,.8,1.2},
  xlabel={$\lambda$ (req/s)},
  ymin=0, ymax=1.60, ytick={0,0.5,1.0,1.5}, ytick pos=left,
  ytick style={orange!85!black, line width=0.35pt},
  yticklabel style={text=orange!85!black},
  enlarge x limits=0.15,
  legend style={evallegend, at={(-0.148,-0.79)}, anchor=north,
    legend columns=2,
    /tikz/every even column/.append style={column sep=4pt}}]
\node[font=\normalsize, text=orange!85!black, anchor=south west]
  at (axis cs:0.12,0.05) {$\Delta_{C2}$ (s)};
\addplot[evalsecondary] coordinates
  {(0.12,1.282)(0.3,1.388)(0.55,1.202)(0.8,0.973)(1.2,0.248)};
\addlegendentry{saved (s)}
\addlegendimage{color=black!72, mark=triangle*, mark size=1.55pt, thick}
\addlegendentry{fills / request}
\end{axis}
\begin{axis}[at={(0.528\columnwidth,0)}, anchor=south west,
  scale only axis, width=0.307\columnwidth, height=0.27\columnwidth,
  symbolic x coords={0.12,0.3,0.55,0.8,1.2}, xtick=\empty,
  ymin=0, ymax=1.6, enlarge x limits=0.15,
  axis y line*=right, axis x line=none, ytick pos=right, ylabel near ticks,
  ytick={0,0.5,1,1.5}, yticklabels={0,0.5,1,1.5},
  ylabel={C2 fills issued / request},
  tick label style={font=\normalsize}, label style={font=\normalsize},
  axis line style={black!55, line width=0.25pt},
  tick style={black!72, line width=0.35pt},
  y label style={text=black!72},
  yticklabel style={text=black!72}]
\addplot[color=black!72, mark=triangle*, mark size=1.55pt, thick] coordinates
  {(0.12,1.38)(0.3,1.39)(0.55,1.29)(0.8,1.20)(1.2,0.89)};
\end{axis}
\end{tikzpicture}
\caption{\wl{Routing} under load, three seeds. (a)~Latency decomposition.
(b)~What C2 saves and the fills it issues.}
\Description{Two panels over five arrival rates. Panel a decomposes
the no-reuse latency into full-system latency and the savings from C1
and C2. Panel b shows C2 latency saving rising from 1.28 to a peak of
1.39 seconds at 0.3 requests per second and then falling to 0.25, while
the speculative fills issued per request fall from 1.38 to 0.89.}
\label{fig:c1c2-ablation}
\end{figure}

\vspace{0.05in}\noindent\textbf{Dynamic adaptation.}
Priced admission lowers complete-trace mean latency to $8.09$\,s,
$5.7\%$ below \textbf{static} and $5.5\%$ below \textbf{always}
(Figure~\ref{fig:dynamic}b). Static opens
candidate generation only with at least three free workers and
queue depth at most two; always opens whenever a candidate exists.
Task accuracy is $0.977$ for \sys{} and static and $0.979$ for
always, with no failures.

The trace brackets a burst (B) and a reuse-heavy phase (C) with
low-load (A) and recovery (D) phases. B and C have the same arrival
rate, but C1 absorbs more work in C. The share of candidate-opening branches rises from
$43\%$ in B to $100\%$ in C, while the break-even threshold
(candidate-set mean, phase median) falls from $0.678$ to $0.380$
(Figure~\ref{fig:dynamic}a).

\begin{figure}[t]
\centering
\begin{tikzpicture}
\begin{axis}[evalaxis, at={(0.150\columnwidth,0)}, anchor=south west,
  scale only axis, width=0.3325\columnwidth, height=0.28\columnwidth,
  title={(a) decisions}, xlabel={phase},
  ylabel={normalized value},
  ytick pos=left,
  symbolic x coords={A,B,C,D}, xtick=data,
  ymin=0, ymax=1.34, ytick={0,0.5,1},
  legend style={at={(1.038,-0.50)}, anchor=north, legend columns=3,
    evallegend,
    /tikz/every even column/.append style={column sep=3pt}}]
\addplot[mark=triangle*, mark size=1.55pt, black!48, thick]
  coordinates {(A,0.574)(B,0.960)(C,0.936)(D,0.698)};
\addlegendentry{load $\bar\rho_H$}
\addplot[color=orange!85!black, mark=square*, mark size=1.25pt, thick]
  coordinates {(A,0.148)(B,0.678)(C,0.380)(D,0.142)};
\addlegendentry{candidate $p^*$}
\addplot[color=blue!75!black, mark=*, mark size=1.45pt, thick, dashed]
  coordinates {(A,1.000)(B,0.427)(C,1.000)(D,1.000)};
\addlegendentry{A0 open}
\end{axis}
\begin{axis}[evalaxis, at={(0.5075\columnwidth,0)}, anchor=south west,
  scale only axis, width=0.3325\columnwidth, height=0.28\columnwidth,
  title={(b) control cost}, xlabel={phase},
  ylabel={latency excess (\%)},
  ytick pos=right, yticklabel pos=right, ylabel near ticks,
  ylabel style={font=\normalsize}, yticklabel style={font=\normalsize},
  axis line style={black!55, line width=0.25pt},
  tick style={black!65, line width=0.35pt},
  symbolic x coords={A,B,C,D}, xtick=data,
  ymin=-5, ymax=27,
  legend style={evallegend, at={(-0.0376,-0.70)}, anchor=north,
    legend columns=3,
    /tikz/every even column/.append style={column sep=4pt}}]
\addplot[evalreference] coordinates {(A,0)(B,0)(C,0)(D,0)};
\addlegendentry{\sys{}}
\addplot[mark=square*, mark size=1.25pt, teal!70!black, thick, dashed,
         error bars/.cd, y dir=both, y explicit,
         error bar style={black!45, thin}] coordinates
  {(A,6.9) +- (0,2.1)(B,6.9) +- (0,1.7)(C,14.1) +- (0,10.0)(D,2.4) +- (0,3.3)};
\addlegendentry{static}
\addplot[mark=*, mark size=1.4pt, orange!85!black, very thick,
         error bars/.cd, y dir=both, y explicit,
         error bar style={black!45, thin}] coordinates
  {(A,0.4) +- (0,1.2)(B,7.9) +- (0,5.0)(C,10.3) +- (0,5.6)(D,2.5) +- (0,4.8)};
\addlegendentry{always}
\end{axis}
\end{tikzpicture}
\caption{Phased adaptation on \wl{Routing}. (a)~What the two levels decide.
(b)~What the simple controls cost.}
\Description{Two panels over phases A through D. The left shows load
peaking in the burst phase, where the candidate break-even threshold
rises to 0.678 and the share of candidate-opening branches falls to 43 percent.
With higher reuse at the same arrival rate, the threshold falls to
0.380 and that share returns to 100 percent. The right shows
the latency excess of the static and always-open controls over \sys{}.}
\label{fig:dynamic}
\end{figure}
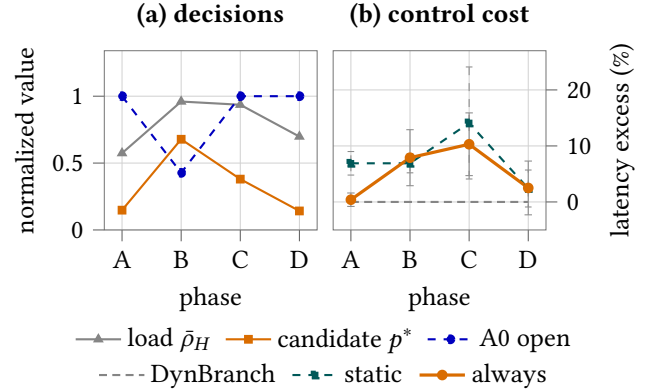

\subsection{Limitations and Future Work}
\label{sec:discussion}

The evaluated workflows use homogeneous serving models, stable branch
coordinates before resolution, and side-effect-free speculative stages.
Benefit depends on recurrence, predictability, available overlap, and
spare capacity. Extensions include templates from emitted
plans~\cite{plancache}, transactional isolation and cross-resource
scheduling for stateful tools and heterogeneous multimodal
pipelines~\cite{murakkab,teola}, and speculative units ranging from
retrieval or prefill~\cite{specrag,predprefetch} to branches.

%% file: sections/related.tex
\section{Related Work}
\label{sec:related}

\noindent\textbf{LLM state and result caches.}
Whole responses are cached under exact or semantic query
matches~\cite{gptcache,instcache,scalm}, plan templates are
reused~\cite{plancache}, prefix or chunk KV state is retained, blended, or
routed~\cite{sglang,preble,characterai,cachedattention,pythia,pbkv,kvflow,cacheblend,cachecraft}, and tool outputs are
memoized~\cite{tvcache,mdpi_hier,cortex}. Continuum preserves KV state across tool-call
pauses using tool-aware lifetimes and program-level
scheduling~\cite{continuum}. These systems retain content or model
state whose identity is already known. \sys{} instead produces work
before the branch outcome exists, then promotes the exact selected
result with its dependency provenance for later reuse.

\vspace{0.05in}\noindent\textbf{Reusable computation and freshness.}
Semantic caching and answering queries using views established
sub-query reuse~\cite{semcache,aquv}; derived state can be maintained
incrementally~\cite{blakeley1986,ivmcount}, as Noria, ReadySet, and
Differential Dataflow do in dataflow
systems~\cite{noria,readyset,diffdataflow}. Freshness scheduling and
self-adjusting computation track when derived work remains
valid~\cite{webviewfresh,sac}. Adaptive query processing and multi-query
optimization revise and share execution plans~\cite{midquery,eddies,mqo}.
\sys{} adopts the same principle (reusable work needs an explicit
dependency and freshness contract) for an LLM subgraph whose dependencies are
external reads, not base tables.

\vspace{0.05in}\noindent\textbf{Agent-aware serving.}
Recent characterization finds heterogeneous resources, long idle
intervals, shifting bottlenecks, and cross-request tool redundancy in
agentic workloads~\cite{agentsysbench}. Agentix treats programs as
first-class scheduling units~\cite{autellix}, while a policy-driven
agent runtime exposes observation, prediction, scoring, and action hooks
between frameworks and engines, evaluating agent-aware KV prefetch and
retention~\cite{agentruntime}. Murakkab configures workflow
components across models and hardware under SLOs~\cite{murakkab}, and
Dyserve adapts per-node model, verifier, and backend choices under
load~\cite{dyserve}.

\vspace{0.05in}\noindent\textbf{Workflow and dataflow optimization.}
Helium couples proactive caching with workflow scheduling, while Halo
plans model residency for known DAGs~\cite{helium,halo}. FlowMesh
deduplicates operator work by content, and Lumilake merges identical
intermediate nodes across known jobs~\cite{flowmesh,lumilake}. Session
DAG and stage-pipeline systems make their inter-stage dependencies
explicit to the scheduler~\cite{parrot,teola,llmcompiler}. Agentix, Alto, and Conveyor add
program-aware preemption, nested-ancestry dataflow orchestration, and streamed
tool/decode overlap~\cite{autellix,alto,conveyor}. \sys{} supplies a
dependency not yet visible to the scheduler and returns the confirmed
result to the same dataflow, so these techniques remain complementary.

\vspace{0.05in}\noindent\textbf{Speculative agents and serving.}
Interactive Speculative Planning, dialogue-response prefetching,
Speculative Actions, DSP, SpecHop, SPORK, and SPEX predict future
actions, reasoning paths, or downstream work and verify the realized
path~\cite{ispplan,dialogueprefetch,specactions,dsap,spechop,spork,spex}.
Skim executes a site-template fast path and verifies its output before
falling back to the full agent~\cite{skim};
Sherlock instead overlaps a produced node's verification with its
downstream work~\cite{sherlock}. Tool-oriented systems
execute predicted calls~\cite{paste,spectool}; PASTE, SPAgent, IdleSpec,
B-PASTE, and DualSpec gate them on confidence, utility, capacity, or
semantic consistency~\cite{paste,spagent,idlespec,bpaste,dualspec}. PASTE is closest: beyond the hold-until-match of
\S\ref{sec:bg:challenges} it co-schedules tools with returning LLM
sessions, but still consumes the confirmed result within its originating
request. Speculate-with-Memory carries predictor
statistics and episodes across tasks~\cite{specmem}, rather than
persisting validated products. Token-level speculative decoding drafts and verifies individual
tokens~\cite{leviathan2023,medusa,eagle,distillspec}, and serving systems
adapt that draft work to load, batch, or
budget~\cite{turbospec,nightjar,adaserve,semiclair,onlinesd}.
A separate family lowers latency by changing what the model computes:
routing to a cheaper model~\cite{routellm,dyserve}, or bounding the
reasoning budget, which can change the answer
itself~\cite{s1scaling,cotfaithful}. \sys{} overlaps candidate subgraphs
with resolving decode; completed results serve later requests after
exact matching and freshness checks. \sys{} reprices each branch at
generation and every fill using live load and observed outcomes.

%% file: sections/conclusion.tex
\section{Conclusion}
\label{sec:conclusion}

\sys{} removes the branch-resolution barrier by treating an unresolved
branch as serving state: a coordinate names candidate work before the
branch resolves, an exact input match with fresh dependencies authorizes
it afterwards, and a two-level controller prices both against live load.
Across four agentic workloads it cuts mean latency by up to $32\%$ over
the strongest prior system and by $46$--$66\%$ against a no-reuse floor,
with workflow results unchanged; with reuse disabled, early fill alone
still removes $8.5$--$47.7\%$. The contract behind that result rests only on a named decision
point and a stage safe to run twice, so it carries to workflows whose
shape is decided at runtime, and priced admission sustains it from idle
to saturated load. Within those two requirements, predictors and pricing
policies can evolve without weakening exact reuse or freshness.

%% file: supplement/appendix.tex
\section{Guarantees and Complexity of Reuse and Admission}
\label{app:optimality}

This appendix derives the guarantees behind \sys{}. Two results
concern C1 and C2: the subgraph key is the coarsest reuse key that is
safe without a model of the stage, so it maximizes the reuse of
Equation~(\ref*{eq:reuse}), and promotion is correspondingly the most
permissive safe rule. For C3, we analyze Action~0's branch-level sign
decision, then establish exact per-candidate break-even, stepwise
optimality, and terminal maximality of Action~1's sequential repricing.
An exact ordered-subset oracle checks the latter.

\paragraph{Setting.}
Fix a realized execution domain $\mathcal E$. An execution $e\in\mathcal E$
of a speculable stage $v(e)$ renders an input $q(e)$,
reads external records at versions $D(e)$, and has semantic input
\[
  t(e)=\bigl(\mathrm{fp}(v(e)),q(e),D(e)\bigr),
\]
where the stage fingerprint covers its code, model, and execution
parameters. It produces $y(e)=F\bigl(t(e)\bigr)$, where $F$ is
deterministic: equal inputs and dependency versions give equal outputs
(\S\ref*{sec:design:lifecycle}). Write $x(e)=\bigl(t(e),z(e)\bigr)$ for
the metadata available before the stage executes, where $z$ is any
further recorded field, such as unrelated parts of the request. A keying
function $\kappa$ is \emph{input-based} if $\kappa(x)$ is a fixed
function of that metadata for every possible $F$: it neither executes nor
models $F$, and cannot inspect the stage's outputs or intermediate
results to decide a match. Every cache that does not model the stage it
serves is input-based in this sense.

Let $\mathcal F$ be the family of stage semantics the serving layer
admits---every deterministic function of $t$, since nothing
further about $F$ is known to it. Call $\kappa$ \emph{uniformly sound} if
for every $F\in\mathcal F$ and all $e_1,e_2\in\mathcal E$, $\kappa(x_1)=\kappa(x_2)$
implies $y(e_1)=y(e_2)$. Write
\[
  \kappa^\star(x)=\bigl(\mathrm{fp}(v),\,h(q),\,D\bigr)
\]
for \sys{}'s reuse name under its freshness guard
(\S\ref*{sec:design:smg}): the stage fingerprint, the rendered input
under a hash $h$ assumed collision-free on inputs in $\mathcal E$, and the
dependency versions the guard certifies.

\begin{lemma}[Soundness]
\label{lem:sound}
$\kappa^\star$ is uniformly sound.
\end{lemma}

\begin{proof}
$\kappa^\star(x_1)=\kappa^\star(x_2)$ gives $t(e_1)=t(e_2)$ because $h$
is collision-free on the realized inputs. Determinism of $F$ yields
$y(e_1)=y(e_2)$.
\end{proof}

\begin{lemma}[Coarsest sound key]
\label{lem:coarsest}
Let $\kappa$ be input-based and uniformly sound. Then
$\kappa(x_1)=\kappa(x_2)$ implies
$\kappa^\star(x_1)=\kappa^\star(x_2)$: the partition $\kappa$ induces
refines the one $\kappa^\star$ induces.
\end{lemma}

\begin{proof}
By contraposition. Suppose $\kappa^\star(x_1)\neq\kappa^\star(x_2)$.
Then $t(e_1)\neq t(e_2)$. Because $\mathcal F$ contains every
deterministic function of $t$, it contains some $F$ with
$F\bigl(t(e_1)\bigr)\neq F\bigl(t(e_2)\bigr)$. Uniform soundness of
$\kappa$ on that $F$ forces $\kappa(x_1)\neq\kappa(x_2)$; because an
input-based $\kappa$ is fixed independently of $F$, it cannot merge the
pair for any other semantics either.
\end{proof}

\begin{corollary}[Maximal hit rate]
\label{cor:hitrate}
Over the same resident entries and subgraph at each decision point $j$,
$\kappa^\star$ maximizes every $P_j^{\mathrm{hit}}$ in
Equation~(\ref*{eq:reuse}) among uniformly sound input-based keys, and
hence maximizes the expected saving $\mathbb{E}[\Delta T_{\mathrm{reuse}}]$.
\end{corollary}

\begin{proof}
A coarser partition serves weakly more demands, and by
Lemma~\ref{lem:coarsest} no uniformly sound input-based key is coarser
than $\kappa^\star$. Soundness is Lemma~\ref{lem:sound}. Each
$T_j^{\mathrm{sub}}$ is fixed, so maximizing every $P_j^{\mathrm{hit}}$
maximizes the sum.
\end{proof}

A whole-response key, which requires the entire request to match, and a
root-anchored prefix key, which requires every token up to $v$ to match,
both separate executions that $\kappa^\star$ merges whenever the
difference lies outside $v$'s rendered input. Each is therefore a strict
refinement on some workflow; Figure~\ref*{fig:ladder} exhibits one, where
they serve $0$ and $2$ of the six taken nodes against $\kappa^\star$'s
$4$. Corollary~\ref{cor:hitrate} says that $4$ is the top of that ladder
and not merely its highest built rung. \emph{Input-based} is necessary
rather than cosmetic: a key that modeled $F$ could merge distinct inputs
with equal outputs, but that lies outside the exact input-identity contract
of \S\ref*{sec:design:smg}.

\hypertarget{supp-c1}{}\begin{corollary}[Completeness of promotion]
\label{cor:complete}
Fix the shadows eligible under C2's lifecycle and C1's freshness and
provenance guards. Promotion---match on $\kappa^\star$ with a fresh
recorded read-set---is the most permissive sound rule over that set: no
demand--shadow pair that some uniformly sound input-based rule could
authorize is refused.
\end{corollary}

\begin{proof}
A uniformly sound input-based rule admitting a pair that $\kappa^\star$
separates would contradict Lemma~\ref{lem:coarsest}. Every remaining
pair has equal keys and is authorized.
\end{proof}

With the soundness argument of \S\ref*{sec:design:lifecycle}, which shows
promotion never serves a value canonical execution would not have
produced, Corollary~\ref{cor:complete} brackets C2 from the other side.

\begin{proposition}[Minimal graph-based write cascade]
\label{prop:cascade}
Assume complete read sets and producer edges, and consider invalidation
rules using only this dependency graph, without modeling stage semantics.
For a write set $W$, let
$A$ contain its direct readers and every descendant that loses a valid
producer. The reverse-index cascade invalidates exactly $A$. If $E_A$
are the traversed producer edges and $R_A$ the dependency references
removed, it takes
$O(|W|+|A|+|E_A|+|R_A|)$ time and $O(|A|)$ auxiliary space. No uniformly
sound rule over this graph can invalidate a strict subset of $A$.
\end{proposition}

\begin{proof}
The reverse index finds exactly the direct readers, and graph traversal
visits each affected node and edge once; removing a node retires each of
its recorded dependencies once. Every node in $A$ lies on a dependency
path from a changed record. For each such node, the admitted semantic
family contains a deterministic function in which that change reaches
its output. A graph-only rule cannot establish that the output is unchanged,
so retaining the node is not uniformly sound. A node outside $A$ either
has no such path or retains a valid producer, so the cascade need not
invalidate it. With
incomplete provenance, \sys{} may extend the cascade
conservatively; soundness remains, but minimality is not claimed.
\end{proof}

\hypertarget{supp-overlap}{}\begin{lemma}[Tight overlap bound]
\label{lem:overlap}
Fix the realized durations at branch $k$. Early execution removes at most
$\min(T_k^{\mathrm{gen}},T_k^{\mathrm{exec}})$ when the selected candidate
is covered and zero otherwise. Expectation over coverage and summation give
Equation~(\ref*{eq:saving}). Equality is attainable when an eligible fill
starts with the resolution window, runs without contention, and is consumed
immediately after resolution.
\end{lemma}

\begin{proof}
Execution cannot overlap more than either the window or its own duration.
The stated schedule attains their minimum; multiplying by the coverage
indicator and taking expectations gives the bound.
\end{proof}

\paragraph{Reuse-path complexity.}
For rendered-input size $|q|$, returned payload size $|y|$, and a read set
of size $r$, content-key construction costs $O(|q|)$, an indexed probe is
expected $O(1)$, deep-copying a hit costs $O(|y|)$, and constructing the
sorted read-set version costs $O(r\log r)$. For $K_b$ shadows and $R_b$
dependency references on branch $b$, its index makes branch resolution and
promotion $O(K_b+R_b)$ rather than linear in all shadows. Exceeding the
CSL's capacity can add an eviction scan over the cache. Aside from payloads, $N$ entries with $R$
dependency references use $O(N+R)$ metadata; caps and TTLs bound this state.

\hypertarget{supp-action0}{}\paragraph{The generation gate (Action 0).}
Action~0 learns one net return per branch
(Equation~(\ref*{eq:plane-target})), charging setup and enabled fills
once and crediting consumption only after a C1 miss. An unproductive
expansion retains its setup charge; mutually exclusive candidates do not
receive separate optimism bonuses. Write
$m_c(x)=x^\top A_c^{-1}q_c$ and
$r_c(x)=\hat\sigma_c\sqrt{x^\top A_c^{-1}x}\ge0$.
After cold trials, unaudited admission opens iff $m_c(x)+r_c(x)>0$.

\begin{proposition}[Conditional Action 0 sign guarantee]
\label{prop:actionzero}
Let $\mu_c(x)=\mathbb E[\hat y_b\mid c,x]$. At a trained, unaudited
decision, $\mu_c(x)\le m_c(x)+r_c(x)$ makes every refusal satisfy
$\mu_c(x)\le0$. If also $\mu_c(x)\ge m_c(x)-r_c(x)$, then
$m_c(x)-r_c(x)>0$ certifies positive conditional net return.
\end{proposition}

\begin{proof}
Refusal means $m_c(x)+r_c(x)\le0$; apply the assumed upper bound.
The lower bound gives the second claim.
\end{proof}

The radius scales with observed target RMS, not a calibrated confidence
bound; Proposition~\ref{prop:actionzero} is conditional. Cold trials and
seeded audits gather feedback outside this sign decision. Pending trials
count toward the cold budget and register a branch token before expansion;
delayed feedback revises that trial's target and $q_c$ without adding
another observation to $A_c$.

\hypertarget{supp-action1}{}\paragraph{The fill gate (Action 1).}
Consider one exclusive branch with candidates $1,\dots,n$ and the
estimates of \S\ref*{sec:design:admission}:
$\tilde p_s=\tilde\pi_s^{\mathrm{sel}}\bar\pi_s^{\mathrm{use}}\in[0,1]$ for useful
consumption, $\hat G_s^{\Delta}\ge0$ for the execution its consumption
saves after C1 lookup, and $C_s=C_s(V_H)\ge0$ for the delay its execution
imposes when it is \emph{not} consumed. At projected state $V_H$, write
$U_s(V_H)=\tilde p_s\hat G_s^{\Delta}-(1-\tilde p_s)C_s(V_H)$.

\begin{proposition}[Stepwise optimality and maximality]
\label{prop:fillgate}
At any fixed $V_H$ with $\hat G_s^{\Delta}+C_s(V_H)>0$, $U_s(V_H)>0$
exactly when $\tilde p_s>C_s(V_H)/(\hat G_s^{\Delta}+C_s(V_H))$,
the threshold $p^\ast$ of Equation~(\ref*{eq:breakeven}).
If both terms vanish, $U_s=0$ and the strict rule rejects the candidate;
the break-even ratio $0/0$ is undefined. Excluding seeded audits,
each Action~1 admission maximizes immediate utility among the candidates
that pass hard feasibility and atomic reservation at the current projected
state. Every admitted fill is therefore value-positive at reservation. If
Action~1 stops because the maximum utility is non-positive, no remaining
candidate can acquire positive utility after any further admissions.
\end{proposition}

\begin{proof}
When $\hat G_s^{\Delta}+C_s(V_H)>0$, rearranging $U_s(V_H)>0$ gives
the threshold; equality gives $U_s(V_H)=0$. In particular,
$C_s=0<\hat G_s^{\Delta}$ gives threshold $0$, and
$\hat G_s^{\Delta}=0<C_s$ gives threshold $1$.
When both vanish, $U_s=0$ directly. The algorithm examines candidates in
decreasing current utility, discards reservation failures, and admits the first feasible
positive one. After an admission, \textsc{ProjectLoad} only increases slot
and KV occupancy. The price curves are componentwise monotone, while
$\tilde p_s$, $\hat G_s^{\Delta}$, and $\hat{\mathbf R}_s$ remain fixed
within the decision; hence every remaining $C_s$ can only rise and every
$U_s$ can only fall. The stopping condition is therefore preserved in all
states reachable by further admissions.
\end{proof}

The proposition characterizes the implemented dynamic process: its next
choice maximizes current feasible utility, and its price-based terminal set
is maximal---no positive-utility fill can be appended. These per-step
properties do not imply maximum total utility over all feasible orders;
Appendix~\ref{app:evidence} compares that total with an exact oracle.

\paragraph{Admission complexity.}
For $K$ candidates and context dimension $d$, Action~0 observation costs
$O(d^2)$; revising a delayed target costs $O(d)$ without changing $A_c$
or the sample count. The direct matrix inversion used to score a context
costs $O(d^3)$ and stores $O(d^2)$ per branch class ($d=6$ here).
Action~1 rescans the
remaining candidates after every projected reservation, taking $O(K^2)$
time and $O(K)$ space.

Appendix~\ref{app:evidence} uses exact dynamic programming over all ordered
subsets in $756$ archived snapshots. Among the $639$ with positive optimum,
the implemented rule selects an exact-oracle set in $626$ ($98.0\%$) and
retains $99.36\%$ of aggregate sequence-aware oracle utility.
With a binding footprint budget, global sequence optimization contains
$0/1$ knapsack even at zero price: for $K$ mutually exclusive candidates,
set $\tilde p_s=1/K$, $\hat G_s^{\Delta}=K v_s$, and KV footprint $w_s$
for each item, with budget $B$ and all other limits slack. Each feasible
subset then has value $\sum_s v_s$ and weight $\sum_s w_s\le B$,
so the problem is NP-hard. For this projected-state update, exact
ordered-subset dynamic programming takes $O(K2^K)$ time and $O(2^K)$ space.

\paragraph{Scope.}
The Action~1 guarantees apply to estimated utility over revealed
candidates. Action~0's sign guarantee concerns its aggregate branch return
under the stated bound; it makes no claim about unrevealed candidates.
The gates' end-to-end latency effect is evaluated in
\S\ref*{sec:eval:c1}--\S\ref*{sec:eval:c1c2-ablation}.

%% file: supplement/evidence.tex
\section{Checks from Archived Runs}
\label{app:evidence}

We use archived artifacts---recorded outputs and admission snapshots, no
new serving runs---to check deterministic execution and to replay the
implemented sequential repricing rule against an exact ordered-subset
oracle. The replay script, input hashes, and per-decision results will
be released with the artifact. The replay holds Action~0 fixed: closed branches reveal no
candidate set for its exact oracle; Figure~\ref*{fig:twolevel}
reports the two-level ablation.

\paragraph{Deterministic execution.}
A three-seed \emph{ReAct} reproducibility run enables deterministic
inference. In its no-reuse
floor, $38$ repeated-question groups contain $96$ executions and $87$
within-group pairs. Every pair has identical final answer text, and
every group has an identical hop count. Across the full and C1-only
arms, hop counts also agree on all $116$ paired question IDs.
All nine request artifacts match their manifest hashes, which supports
Lemma~\ref{lem:sound}'s determinism premise in this configuration.

\paragraph{Sequential-admission replay.}
We use all $756$ admission snapshots from a three-seed \emph{Routing} run,
including nine startup warmup events. The remaining
$747$ events join to the request artifacts by workflow ID. Each snapshot
contains two or three candidates, totaling $2{,}133$, with the probability,
gain, cost, footprint, and load values used by the gate. The applicable
price functions and curves match the archived source; reconstructed
costs agree with all logged costs within their printing precision.

For each snapshot, an exact dynamic program optimizes over all ordered
subsets under recorded KV headroom. Projected state depends only on the
selected subset's cardinality and total footprint, so one best prefix per
subset suffices. Both oracle and implementation use the archived estimates
and production price curves, project slot and KV occupancy after each
admission, and score $(s_1,\ldots,s_m)$ by
$J_{\mathrm{seq}}=\sum_i U_{s_i}(V_{i-1})$. Seeded audits are disabled.

\renewcommand{\thetable}{\thesection.\arabic{table}}
\begin{table}[H]
\centering
\begin{tabular}{lr}
\toprule
Measure & Repriced / exact oracle \\
\midrule
Oracle-set match & $626/639$ ($98.0\%$) \\
$\sum J_{\mathrm{seq}}/\sum J_{\mathrm{seq}}^\star$ & $99.36\%$ \\
Median $J_{\mathrm{seq}}/J_{\mathrm{seq}}^\star$ & $99.86\%$ \\
5th percentile & $97.35\%$ \\
Minimum & $94.32\%$ \\
\bottomrule
\end{tabular}
\caption{Exact replay of the sequence-aware admission objective.
All rows use the $639$ snapshots with positive optimum; the other $117$
have zero optimum and both methods select nothing.}
\label{tab:admission-replay}
\end{table}

Among the $639$ positive-optimum snapshots, the implemented order attains
the exact sequence value in $281$. Another $345$ select an oracle-optimal
set in a lower-value order; only $13$ differ in selected set. The rule
retains $99.36\%$ of aggregate oracle utility
(Table~\ref{tab:admission-replay}).
All candidates together use at most $16.8\%$ of recorded KV headroom, so
the footprint-cap knapsack case does not arise in this archive. This replay
gives an exact comparison for the logged estimates and revealed candidates;
\S\ref*{sec:eval:c1}--\S\ref*{sec:eval:c1c2-ablation} measure the end-to-end latency effect.